\documentclass[aps,prx,floatfix,nofootinbib,twocolumn,superscriptaddress,longbibliography]{revtex4-2}

\usepackage{graphicx}
\usepackage{dcolumn}
\usepackage{bm}
\usepackage{natbib}
\usepackage{amsmath}
\usepackage{subfigure}
\usepackage{dsfont}
\usepackage{eufrak}
\usepackage{tikz}
\usepackage{appendix}
\usepackage{makecell}
\usepackage[margin=0.8in]{geometry}
\usepackage[normalem]{ulem}
\usepackage{braket}
\usepackage{amsthm}
\usepackage{amssymb}
\usepackage[scr=esstix]{mathalpha}
\usepackage[bottom]{footmisc}
\usepackage[normalem]{ulem}
\usepackage{comment}

\newcommand{\n}{\notag \\}

\newcommand{\half}{\frac{1}{2}}

\newcommand{\tr}{\text{tr}}
\newcommand{\dd}{\text{d}}

\newcommand{\Wg}{\text{Wg}}

\newcommand{\tmoment}{\widehat{\Phi}}
\newcommand{\tmomentd}{\widehat{\Phi}^{\dagger}}
\newcommand{\rqc}{\text{RQC}}

\newcommand\addtag{\refstepcounter{equation}\tag{\theequation}}

\newtheorem{lemma}{Lemma}
\newtheorem{theorem}{Theorem}
\newtheorem{conjecture}{Conjecture}

\theoremstyle{definition}
\newtheorem{definition}{Definition}

\long\def\/*#1*/{}

\begin{document}

\preprint{APS/123-QED}

\title{Approximate t-designs in generic circuit architectures}

\def\urbana{
Institute for Condensed Matter Theory and IQUIST and NCSA Center for Artificial Intelligence Innovation and Department of Physics, University of Illinois at Urbana-Champaign, IL 61801, USA
}

\def\chicago{
Department of Computer Science, The University of Chicago
}

\author{Daniel Belkin}
\thanks{Co-first authors.}
\affiliation{\urbana}
\author{James Allen}
\thanks{Co-first authors.}
\affiliation{\urbana}
\author{Soumik Ghosh}
\affiliation{\chicago}
\author{Christopher Kang}
\affiliation{\chicago}
\author{Sophia Lin}
\affiliation{\chicago}
\author{James Sud}
\affiliation{\chicago}
\author{Fred Chong}
\affiliation{\chicago}
\author{Bill Fefferman}
\affiliation{\chicago}
\author{Bryan K. Clark}
\email{bkclark@illinois.edu}
\affiliation{\urbana}
    
\date{\today}

\begin{abstract}
Unitary \(t\)-designs are distributions on the unitary group whose first \(t\) moments appear maximally random. 
Previous work has established several upper bounds on the depths at which certain specific random quantum circuit ensembles approximate \(t\)-designs. Here we show that these bounds can be extended to any fixed architecture of Haar-random two-site gates. This is accomplished by relating the spectral gaps of such architectures to those of 1D brickwork architectures.
Our bound depends on the details of the architecture only via the typical number of layers needed for a block of the circuit to form a connected graph over the sites.  When this quantity is bounded, the circuit forms an approximate $t$-design in at most linear depth. We give numerical evidence for a stronger bound which depends only on the number of connected blocks into which the architecture can be divided. We also give an implicit bound for nondeterministic architectures in terms of properties of the corresponding distribution over fixed architectures.

\end{abstract}

\maketitle

\renewcommand{\thesection}{\Roman{section}}
\newcommand{\np}{k}

\section{Introduction}
Random quantum circuits are an important tool in the study of natural and engineered quantum systems. In quantum computing, random circuits have been suggested for randomized benchmarking \cite{Emerson2005a,Magesan2012}, security, and state preparation \cite{Hayden2003}. Recent claims of quantum supremacy have hinged on the hardness of classical simulation of random circuits (\cite{boixo2016characterizing,Bouland2019a, ac, BFLL, movassagh,KMM21}). Random circuits have also been proposed as models for information scrambling in black holes \cite{Hayden2007,Sekino2008}, and more general random tensor networks have been used as an explicit construction of the holographic duality in ADS/CFT \cite{Hayden2016}. In quantum information theory, random circuits are the standard setting in which to study measurement-induced phase transitions\cite{Skinner2019}, and they are used as an analytically-tractable model of quantum ergodicity and chaos \cite{Liao2022}. Random circuit models also serve as an interesting theoretically tractable model for more complicated realistic physical systems. Their maximally generic dynamics are often a valuable source of intuition as to what one should expect under typical time-evolution.

The geometric structure of a quantum circuit plays a critical role in the flow of information. While initial work on quantum computing focused on one-dimensional architectures, 2D layouts, such as the Sycamore processor used in the quantum advantage experiment of Ref.~\onlinecite{Arute2019}, have become increasingly popular. Recent work has explored modular architectures, in which fully connected nodes are sparsely connected with each other \cite{Pirker2018}. Circuit models for physical systems often require geometric locality on some two- or three-dimensional lattice. In certain cases, such as the SYK model, the interactions are instead all-to-all. Meanwhile, circuit architectures with the connectivity pattern of a tree or Multiscale Entanglement Renormalization Ansatz \cite{Vidal2008} are a natural setting in which to study holography. Such circuits may also be useful for robust quantum simulations of many-body systems \cite{Kim2017}. 

In the limit of large depth, a sufficiently well connected random quantum circuit will eventually scramble quantum information \cite{Emerson2005}. Perhaps the most basic question about random quantum circuits is the rate of this scrambling.  This is quantified by the \textit{approximate t-design depth} \cite{Ambainis2007} , which captures the depth at which an observer with access to at most $t$ measurements can no longer reliably distinguish the circuit from a random global unitary.

Prior work has given bounds on the t-design depths for a few special classes of architectures. However, it was not previously clear how the rate of information scrambling depended on the spatial structure of the circuit. In particular, one might have expected irregular or modular architectures to give qualitatively different behavior, e.g. exponentially slow convergence. 
The main goal of this work is to show that any reasonable architecture forms an approximate \(t\)-design in linear depth. We will bound the rate of convergence in terms of the connectedness of the architecture.

\subsection{Prior work}
The Haar measure on the unitary group is clearly a fixed point of any quantum circuit distribution, since it is invariant under any unitary gate. Ref.~\onlinecite{Emerson2005} shows that random circuits satisfying a universality condition converge to this fixed point in the limit of large depth. But this uniform convergence is very slow, requiring a circuit depth that scales exponentially with system size $N$. On the other hand, the expected values of specific observables sometimes approach their Haar values much faster (e.g. in depth $\log N$) \cite{Dalzell2022}. But this fast convergence depends on specific details of the observables considered and is not necessarily universal for other quantities of interest. The approximate $t$-design depth was first introduced in ref.~\onlinecite{Ambainis2007} as an intermediate measure of convergence. It is strong enough to guarantee convergence of any experimentally observable property, but occurs much faster than uniform convergence of measure.

Much of the prior work on approximate $t$-design depths has focused on the 1D brickwork architecture. Ref.~\onlinecite{Brandao2016} showed that the approximate \(t\)-design depth in this case was at most \(O(t^{9.5 + o(1)}N)\). For local Hilbert space dimension \(q = 2\), ref.~\onlinecite{Haferkamp2022} tightened this bound to \(O(t^{5 + o(1)} N)\). Ref.~\onlinecite{Hunter-Jones2019} gave a mapping to a statistical-mechanical model of interacting domain walls and used it to establish tighter bounds when either \(t = 2\) or \(q \rightarrow \infty\). 
Ref.~\onlinecite{Harrow2023} extended this work to a particular family of $D$-dimensional brickwork architectures. In the limit of small \(\epsilon\) and large \(N\), they established that the approximate \(t\)-design depth scales as at most \(O(N^{1/D})\) (although the dependence on $t$ remains an open question).

The other class of prior work focuses on what we term \textit{nondeterministic architectures}, in which the spatial structure of the architecture is also random. Typically gate locations are assumed to be drawn independently and identically from the uniform distribution over the edges of some graph over the sites.
In this context ref.~\cite{Ambainis2007} established an approximate \(2\)-design size of at most \(O(N^2)\) gates for the all-to-all graph. Ref.~\onlinecite{Brandao2016} found \(O(t^{9.5}N^2)\) for the linear graph, which ref.~\onlinecite{Oszmaniec2022} extended to \(O(t^{9.5} \log^4(t) N^3)\) for any graph which admits a Hamiltonian path. Ref.~\onlinecite{Mittal2023} develops an alternative strategy which yields a bound of the form \(N^{O(\log N)} \text{poly}(t)\) for arbitrary graphs. In addition, they give a bound of the form \(O(|E|N \text{poly}(t))\) for graphs with \(|E|\) edges, bounded degree, and bounded spanning-tree height. For certain bounded \(t\) and degree, the requirement of bounded spanning-tree height can be relaxed. 

\subsection{Summary of results}
 We obtain bounds on the $t$-design depth for all architectures, with stronger bounds if the architecture satisfies certain properties. First, consider a periodic architecture composed of \(\ell\) complete layers of \(2\)-site gates on $N$ qubits. For large \(t\), our bound for the \(\epsilon\)-approximate \(t\)-design depth becomes
\[d_* = t^{(5 + o(1))(\ell-1)} \left(2N t \log 2 + \log \frac{1}{\epsilon}\right) \addtag\]

We can relax each of these assumptions to obtain looser bounds for larger classes of circuits. The most general result covers an architecture with local Hilbert space dimension \(q\) which may not be periodic or consist of complete layers.\footnote{See Definition \ref{def:complete} or Figure \ref{fig:circuit_classification}}
We partition such a circuit into blocks of layers such that the gates in each block form a connected graph over all the sites. Let $\bar{\ell}$ be the average number of layers per block.\footnote{For architectures without any regular structure, \(\bar{\ell}\) may itself depend on depth, so this is not always an explicit formula for the \(t\)-design depth. See Theorem \ref{thm:aperiodic} for a precise definition of \(\bar{\ell}\).}
For these architectures, our bound is of the form
\begin{equation}
\label{eq:general_critical_depth}
d \geq t^{(15.2 + o(1))(\bar{\ell}-1)} (2N t \log q + \log \frac{1}{\epsilon})
\end{equation}
We also give numerical evidence for two conjectures under which Equation \ref{eq:general_critical_depth} can be strengthened to 
\[d_* =  \frac{\left(2Nt \log q + \log \frac{1}{\epsilon}\right)\bar{\ell}}{2\log \frac{q^2 + 1}{2q}} \addtag\]
For nondeterministic architectures,  we obtain an implicit bound in terms of the joint distribution of the effective number of layers in the circuit and the number of connected blocks.

\subsection{Structure of the proof}
We wish to prove that the distribution induced by a random circuit architecture approaches the Haar measure. Following previous work on approximate t-designs \cite{Brown2010}, we begin by expressing the frame potential as a tensor network of single-gate moment operators (Section.~\ref{section:tn_picture}). In Section~\ref{section:finding_gap}, we focus on the case of periodic complete architectures and show that the \(t\)-design depth is determined by the spectral gap of the transfer matrix $T$. (The assumptions of completeness and periodicity will be relaxed in sections~\ref{section:summary}B and C, respectively.) 

We then decompose the transfer matrix into a product of \(\ell\) layers of gates. Each layer is an orthogonal projection operator. We wish to bound the spectral gap of the transfer matrix in terms of the geometry of the subspaces to which the layers project. The key insight at this stage is to consider the way adding a new layer shrinks the unit eigenspace of the product, reducing the norm of the excluded vectors. In Section~\ref{section:layer_gap_bound}, we bound the spectral gap in terms of these norm reductions. The impact of a new layer on the unit eigenspace can be represented by a graph of nodes and edges, which we term the \textit{cluster-merging picture}. 

The next step is to simplify the cluster-merging picture at each layer. In Section~\ref{section:cluster_splitting}, we show that you can unravel the cluster-merging graph of each layer into a collection of loops without increasing the spectral gap. This is useful because each loop is the transfer matrix of a periodic 1D brickwork architecture. At this stage we have obtained a lower bound on our transfer matrix spectral gap in terms of the spectral gaps of 1D brickworks.

In Section~\ref{section:1d_solution}, we proceed to show the 1D brickwork spectral gap itself may be bounded by previous results on 1D approximate \(t\)-design depths. Many of these results were actually originally proven in terms of the spectral gap, but our argument applies even to bounds obtained by other methods. Together these steps allow us to turn a bound on the 1D \(t\)-design depth into a bound for generic architectures. 

We also discuss extensions of our techniques to other architectures. We show how our bounds can be applied in expectation to architectures in which gate locations are drawn randomly, giving an implicit bound in terms of properties of the induced distribution over fixed architectures. We observe that our techniques may be adapted to give tighter bounds for architectures with special structure and give an explicit example for the case of higher-dimensional brickwork architectures.

\subsection{Definitions}
A quantum circuit on \(N\) sites of local dimension \(q\) corresponds to a unitary \(U_c \in U(q^{N})\). A random quantum circuit architecture then induces a measure $\varepsilon_C$ on \(U(q^N)\). Define the associated \(t\)-fold  channel
\begin{gather}
    \Phi_\text{RQC}(\rho) = \int_{\varepsilon_C} U_C^{\otimes t} \rho (U_C^\dagger)^{\otimes t} \dd U_C
\end{gather}
\begin{definition}
    An \(\epsilon\)-approximate unitary \(t\)-design\cite{Brandao2016} is a measure \(\varepsilon_C\) on \(U(q^N)\) such that the diamond norm distance between the corresponding \(t\)-fold channel $\Phi_\text{RQC}$ and that of the Haar measure is at most $\epsilon$:
    \begin{gather}
        ||\Phi_\text{RQC} - \Phi_\text{Haar}||_\diamond \leq \epsilon
    \end{gather}
    We will often shorten ``\(\epsilon\)-approximate unitary \(t\)-design'' to ``\((\epsilon,t)\)-design''.
\end{definition}

\begin{definition}
    \label{def:complete}
    We call an \(L\)-layer random circuit architecture on \(N\) sites \textit{complete} if  each of the \(L\) layers consists of \(\frac{N}{2}\) Haar-random \(2\)-site unitary gates. In other words, exactly one gate acts on each site per layer.
\end{definition}

\begin{definition}
    An \textit{\(\ell\)-layer periodic} random circuit architecture repeats the layout of its layers with period \(\ell\) as the depth increases. Note that gates themselves are independently random at every depth; only the spatial arrangement of the gates is repeated.
\end{definition}

\begin{definition}
\label{def:denseblock}
A \textit{connected $\ell$-block} of a circuit architecture on \(N\) sites is a contiguous block of \(\ell\) layers such that the gates in the block form a connected graph over all $N$ sites.
\end{definition}
These definitions are illustrated in Figure \ref{fig:circuit_classification}.

\begin{figure*}
    \centering
    \begin{tikzpicture}
        \begin{scope}
            \node[anchor=north west] (image_a) at (0,0)
            {\includegraphics[width=1.2\columnwidth]{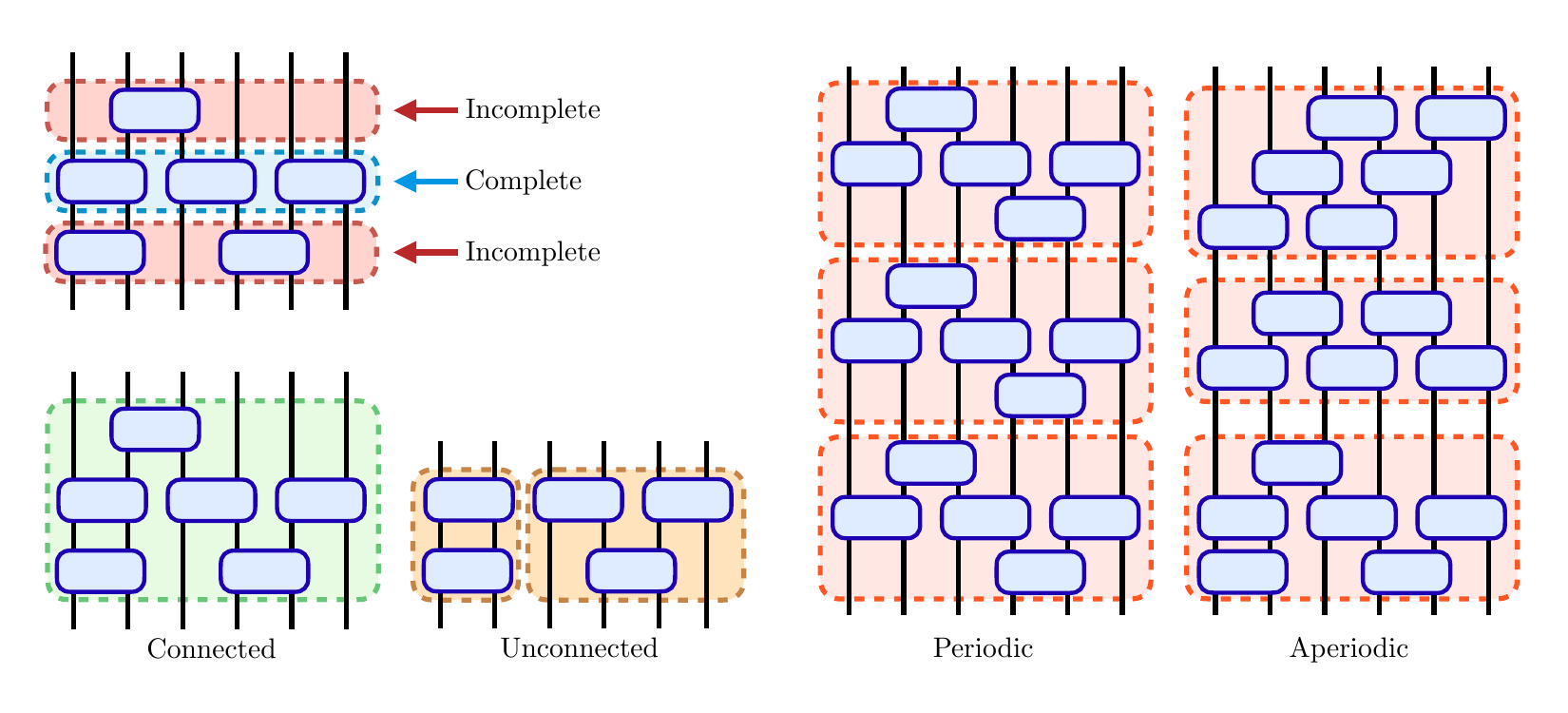}};
            \node [anchor=north west] (note) at (-0.05,-0.4) {\small{\textbf{a)}}};
            \node [anchor=north west] (note) at (-0.05,-2.4) {\small{\textbf{b)}}};
            \node [anchor=north west] (note) at (4.95,-0.4) {\small{\textbf{c)}}};
        \end{scope}
    \end{tikzpicture}
    \vspace*{-0.5cm}
    \caption{Different types of layers, blocks and architectures. (a) The middle layer is complete, since every site is acted on by exactly one gate. The upper and lower layers are both incomplete. (b) The left-hand block is connected while the right-hand block is made up of two unconnected components. (c) The left-hand architecture is periodic, while the right-hand is aperiodic. Both architectures contain three connected blocks.
    }
    \label{fig:circuit_classification}
\end{figure*}
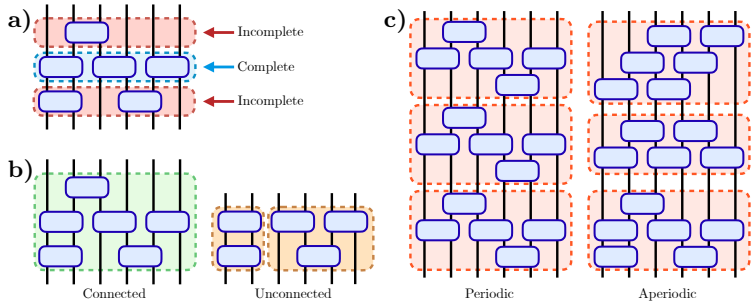

Our most general results are for complete architectures and depend on the frequency and size of connected $\ell$-blocks. However, we obtain a more explicit form for the special case of periodic architectures. We obtain each result by reduction to the well-studied 1D brickwork architecture.
\begin{definition}
    The \textit{\(N\)-site 1D brickwork} is a complete \(2\)-layer periodic random circuit architecture, equipped with an ordering $1...N$ of the sites, such that the first layer applies gates on sites \(\{2j, 2j+1 \mod N\}\), while the second layer applies gates on sites \(\{2j-1 \mod N,2j\}\). In this paper, the spatial boundary conditions are periodic unless otherwise specified.
\end{definition}

\subsection{Main theorems}
\begin{theorem}
\label{thm:periodic_complete}
Suppose that the \(N\)-site 1D brickwork architecture (with either open or periodic boundary conditions) forms an \(\epsilon\)-approximate \(t\)-design after at most
\begin{gather}
    \label{eq:1d_bound_template}
    k_1 = C(N,q,t) \log \frac{1}{\epsilon} + o_\epsilon \left(\log \frac{1}{\epsilon}\right)
\end{gather}
periods. Then arbitrary complete \(\ell\)-layer periodic architectures form an \(\epsilon\)-approximate \(t\)-design after at most
\begin{gather}
    \label{eq:complete_bound}
    k_* = \frac{2Nt \log q + \log \frac{1}{\epsilon}}{\log \frac{1}{s_*}}
\end{gather}
periods, where
\[s_* = 1 - \left(1 - \exp\left(-\frac{1}{2C(q,t)}\right)\right)^{\ell - 1}\addtag\]
and \(C(q,t) = \sup_N C(N,q,t)\). We may relax the bound to the more legible
\begin{gather}
    k_* = \big(4\overline{C}(q,t)\big)^{\ell-1} \left(2Nt \log q + \log \frac{1}{\epsilon}\right)
\end{gather}
by defining $\overline{C}(q,t) \equiv \max\big(C(q,t), \frac{1}{2}\big)$
\end{theorem}

\begin{theorem}
\label{thm:periodic_incomplete}
If we do not require that the layers be complete, an \(\ell\)-layer periodic architecture forms an \(\epsilon\)-approximate \(t\)-design after at most
\begin{gather}
    \label{eq:incomplete_bound}
    k_* = \frac{2Nt \log q + \log \frac{1}{\epsilon}}{\log \frac{1}{s_*}}
\end{gather}
periods, where
\[s_* = 1 - \left(1 - \exp\left(-\frac{1}{2C(\sqrt{q},t)}\right)\right)^{\ell - 1}\addtag\]
and \(C(\sqrt{q},t)\) is defined as in Theorem \ref{thm:periodic_complete}.
\end{theorem}

We also show an alternative result on incomplete layers in Section~\ref{subsection:incomplete_proof} which depends on $C(q,t)$ instead of $C(\sqrt{q},t)$, but at the cost of multiplying \(\ell\) by a factor \(O(\log \log N)\). We also show that the architecture need not actually be periodic.

\begin{theorem}
\label{thm:irregularly_connected}
The results of Theorems~\ref{thm:periodic_complete} and \ref{thm:periodic_incomplete} hold even if the architecture is not periodic, with \(\ell\) replaced by an ``average connection depth'' defined formally in Theorem \ref{thm:aperiodic}.
\end{theorem}

We will begin with a proof of the periodic, complete case. This proof is simpler and illustrates the essential elements of our strategy. The argument will then be extended to the incomplete and aperiodic cases.

Finally, we show that conditional on two conjectures, we can omit the dependence on \(\ell\) entirely to obtain a much simpler result.
\begin{theorem}
    \label{thm:best_guess}
    Suppose Conjectures \ref{conjecture:connection_count} and \ref{conjecture:s_exact_open} hold. Then any architecture which can be divided into 
    \begin{equation}
        k_* = \frac{2Nt \log q + \log \frac{1}{\epsilon}}{2 \log \frac{q^2 + 1}{2q}}
    \end{equation}
    connected blocks forms an \(\epsilon\)-approximate \(t\)-design. 
\end{theorem}

\subsection{Known values of \(C(q,t)\)}
\label{subsection:cvals}
Previous works imply the following bounds on \(C(q,t)\):
\begin{itemize}
    \item For general parameters, the best known bound is that of ref.~\onlinecite{Brandao2016}.
    For the case of 1D brickwork circuits\footnote{Note that our bound is based on a fixed 1D brickwork instead of a random layer ordering. This allows us to tighten the bound of ref.~\cite{Brandao2016} for parallel local random circuits by a factor of \(2\), by skipping the factor of $\frac{1}{2}$ relaxation of $\lambda_2$ in their Lemma 21. } where \(q \geq 2\), the authors give
    \begin{gather}
        C(q,t) = 261500\lceil\log_q (4t)\rceil^2 q^2 t^{5 + \frac{3.1}{\log q}}
    \end{gather}
    For integer \(q^2 \geq 2\), the more general form is
    \begin{multline}
        C(q,t) = 234 (q^2+1) e^{\frac{2.5 \log(4) (1+\log(q^2+1))}{\log (q)} + 1}\\
        \times \lceil\log_q (4t)\rceil^2  t^{5 + \frac{5(1 + \log(1 + q^{-2}))}{2\log q}} 
    \end{multline}
    When $q \geq 2$, we can replace the $t$-exponent with $9.5 \geq 5 + \frac{\log(e(1 + q^2))}{\log q}$. Similarly, for \(q^2 \geq 2\) we have an exponent of at most \(15.2\).
    
    \item For \(q = 2\), Equation 22 of ref.~\onlinecite{Haferkamp2022} gives the tighter bound:
    \[k_1 = \alpha \log^5(t) t^{4 + \frac{3}{\sqrt{\log_2 t}}} \left(2 N t + \log_2 \frac{1}{\epsilon}\right) \addtag\]
    where \(\alpha = 10^{13}\). This gives 
    \[C(2,t) = \frac{\alpha}{2 \log 2} \log^5(t) t^{4 + \frac{3}{\sqrt{\log_2 t}}}\addtag\]
    
    \item For \(t = 2\) and any \(q > 1\), Equation 27 of  ref.~\onlinecite{Hunter-Jones2019} gives 
    \[C(q,2) = \left(2\log \frac{q^2 + 1}{2q}\right)^{-1}\]
    for 1D brickwork circuits with open boundary conditions.\footnote{There is a factor-of-two difference between our definition of \(C\) and that of ref.~\onlinecite{Hunter-Jones2019}. This is because we are counting periods instead of layers.} We show in Appendix \ref{app:1d_t2_bound} that periodic boundary conditions improve the bound to 
    \[C(q,2) = \left(4\log \frac{q^2 + 1}{2q}\right)^{-1} \addtag\] 
    
    \item In the limit \(q \rightarrow \infty\), Equation 36 of  ref.~\onlinecite{Hunter-Jones2019} shows that the leading-order term is \(C(q,t) =  \left(2 \log \frac{q}{2}\right)^{-1}\)
    with open boundary conditions. Periodic boundary conditions again tighten\footnote{There can be no single domain wall in periodic boundary conditions, so the leading term is now given by the two-domain-wall sector.}
    this to
    \[C(q,t) =  \left(4 \log \frac{q}{2}\right)^{-1} + o_q(\log^{-1}q)\addtag\]
    
    \item Following the conjecture of ref.~\onlinecite{Hunter-Jones2019}, we suspect that the sharp bound is
    \begin{equation}
        C(q,t) = \left(4\log \frac{q^2 + 1}{2q}\right)^{-1}
        \label{eq:C_exact_periodic}
    \end{equation}
    This is analogous to ref.~\onlinecite{Hunter-Jones2019}'s conjecture for the open-boundary case. Numerical evidence is given in Appendix \ref{app:brickwork_t_independence}. \\
\end{itemize}

\section{Approximate $t$-designs and tensor network picture} \label{section:tn_picture}
The first phase of our proof follows the standard reduction from approximate \(t\)-designs to a tensor network of averaged gates\cite{Dalzell2022}. For the sake of completeness and notational clarity, sections \ref{section:tn_picture}-\ref{section:finding_gap} outline the key steps. 

For a random quantum circuit channel $\Phi_{\text{RQC}}$ formed from a circuit ensemble $U_C \in \varepsilon_C$, the diamond norm difference from the Haar distribution is bounded in terms of the frame potential\cite{Hunter-Jones2019}:
\begin{gather}
    ||\Phi_\text{RQC} - \Phi_\text{Haar}||_\diamond^2 \leq q^{2Nt}\left(\mathcal{F}^{(t)}_{\varepsilon_C} - \mathcal{F}^{(t)}_{\text{Haar}}\right) \label{eq:hunter-jones}
\end{gather}
\begin{gather}
    \label{eq:define_frame_potential}
    \mathcal{F}^{(t)}_{\varepsilon_C} = \int_{\varepsilon_C^{\otimes 2}} |\tr(U_C^\dagger V_C)|^{2t} \dd U_C \dd V_C.
\end{gather}
Both the random quantum circuit channel and the frame potential can be written in terms of the $t$-th moment operator
\begin{gather*}
    \tmoment_\text{RQC} = \int_{\varepsilon_C} U_C^{\otimes t,t} \dd U_C 
\end{gather*}
where $U_C^{\otimes t,t} \equiv U_C^{\otimes t} \otimes (U_C^*)^{\otimes t}$. This is a matricization of the quantum channel $\Phi_{\text{RQC}}$, i.e. 
\[\tmoment_\rqc\left(\text{vec}(\rho)\right) = \text{vec} \left(\Phi_\rqc(\rho)\right)\addtag\]
We also have
\begin{align}
    \mathcal{F}^{(t)}_{\varepsilon} 
    &= \int_{\varepsilon_C^{\otimes 2}} \tr(U_C^{\dagger \otimes t,t} V_C^{\otimes t,t})\dd U_C \dd V_C\\
    &= \tr(\tmomentd_\rqc \tmoment_\rqc) \label{eq:frame_potential}
\end{align}
\begin{figure}
    \centering
    \begin{tikzpicture}
        \begin{scope}
            \node[anchor=north west,inner sep=0] (image_a) at (0.4,0)
            {\includegraphics[width=0.85\columnwidth]{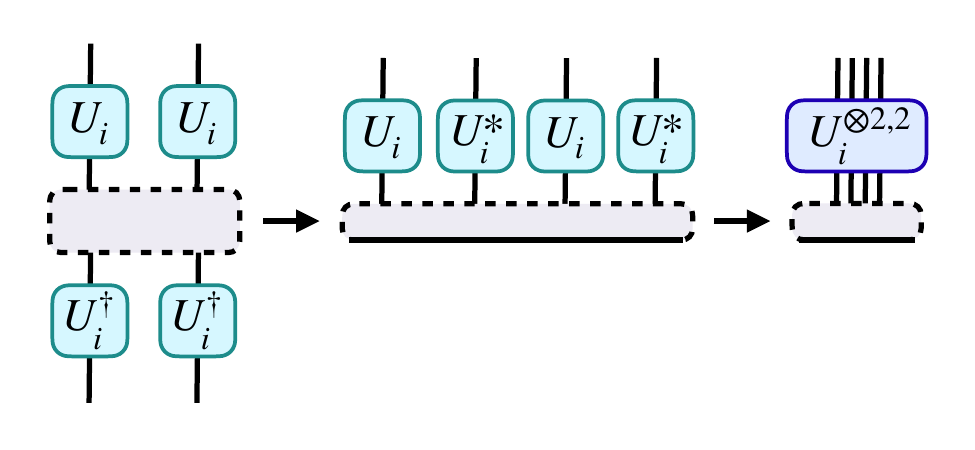}};
            \node [anchor=north west] (note) at (-0.6,-0.1) {\small{\textbf{a)}}};
        \end{scope}
        \begin{scope}
            \node[anchor=north west,inner sep=0] (image_b) at (-0.7,-3)
            {\includegraphics[width=1.05\columnwidth]{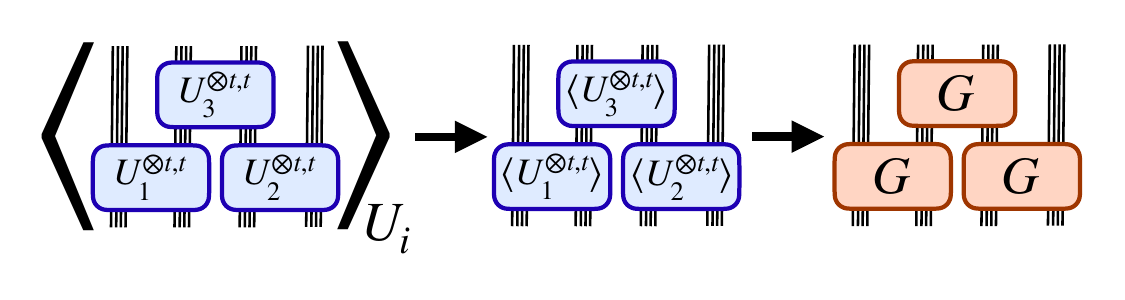}};
            \node [anchor=north west] (note) at (-0.6,-3) {\small{\textbf{b)}}};
        \end{scope}
    \end{tikzpicture}
    \vspace*{-1.2cm}
    \caption{(a) Folding all $t$ copies of a specific gate $U_i$ and $U_i^\dagger$ in the channel $\Phi_\text{RQC}$ into a single operator $U_i^{\otimes t,t}$, where $t=2$. The gray dotted region indicates a density matrix on which the channel acts. This is vectorized in the middle panel, forming a matricization of the $2t$ copies of $U_i$. (b) The reduction of the $t$-th moment operator $\tmoment_\rqc$ to a network of identical projection tensors $G$. Each unitary $U_i^{\otimes t,t}$ can be averaged separately into an independent copy of $G$.}
    \label{fig:moment_generator}
\end{figure}
We assume that $U_C$ consists of two-site unitary gates $U_i$ drawn independently from the Haar distribution over \(U(q^2)\). Since distinct gates are independent, we can average over each gate separately (Figure~\ref{fig:moment_generator}). The averaging joins the \(2t\) copies of each gate that appear in \(U_C^{\otimes t,t}\) into a single operator $G$. The action of \(G\) depends only on \(t\) and the number of sites on which $U_i$ acts. $\tmoment_\rqc$ becomes the contraction of a tensor network in the shape of the original circuit $U_C$, but with each $U_i^{\otimes t,t}$ replaced with its average $G$.

The individual $G$'s can be written in terms of single-site \textit{permutation states}. Given a permutation \(\sigma \in S_t\), we define a particular maximally-entangled state on $t$ pairs of sites
\begin{gather}
    |\sigma\rangle = q^{-t/2}\sum_{\vec{i} \in \mathbb{Z}_q^t} |\vec{i}\rangle |\sigma(\vec{i})\rangle
\end{gather}
We will call a tensor power of a permutation state on on \(m\) sites $|\sigma\rangle^{\otimes m} = |\sigma\rangle |\sigma \rangle ... |\sigma \rangle$ a \textit{uniform permutation state}.
\begin{theorem}
    \label{thm:projector}
    Let \(U\) be an \(m\)-site Haar-random unitary. Let \(G\) be the expected value of the corresponding moment operator $U_i^{\otimes t,t}$. Then \(G\) is a projector onto the space spanned by the uniform permutation states $|\sigma\rangle^{\otimes m}$.
\end{theorem}
A proof may be found in ref.~\onlinecite{Brown2010}. In particular, we see that the $t$-th moment operator for the Haar distribution over all $N$ sites is just the orthogonal projector on to the globally uniform permutation states on all \(N\) sites. Furthermore, these states also span the unit eigenspace of the moment operator of any architecture: 
\begin{lemma}
    \label{lemma:global_is_unique}
    If the circuit architecture is connected, the unit eigenspace of $\tmoment_\rqc$ is spanned by the globally uniform permutation states $\ket{\sigma}^{\otimes N}$. 
\end{lemma}
\begin{proof}
The support of the distribution over the unitaries induced by a random architecture is a universal gate set if and only if the architecture is connected.\footnote{The connectedness is relevant because it allows the use of SWAP gates to obtain any pairwise operation.} \cite{Barenco1995} Suppose we apply \(k\) repetitions of the circuit architecture. As \(k \rightarrow \infty\), ref.~\onlinecite{Emerson2005} shows that the induced measure on \(U(q^N)\) converges to the Haar measure. It follows that the corresponding moment operator \(\tmoment_\rqc^k\) converges to that of the Haar measure, which is the projector onto the span of the uniform permutation states. But since \(\tmoment_\rqc\) is norm-nonincreasing, \(\lim_{k \rightarrow \infty} \tmoment_\rqc^k\) is a projector on to the unit eigenspace of \(\tmoment_\rqc\), so the unit eigenspace of \(\tmoment_\rqc\) must be the same as that of \(\tmoment_\text{Haar}\). Theorem \ref{thm:projector} completes the argument.
\end{proof}

\begin{figure}
    \centering
    \begin{tikzpicture}
        \begin{scope}
            \node[anchor=north west,inner sep=0] (image_a) at (0,0)
            {\includegraphics[width=0.99\columnwidth]{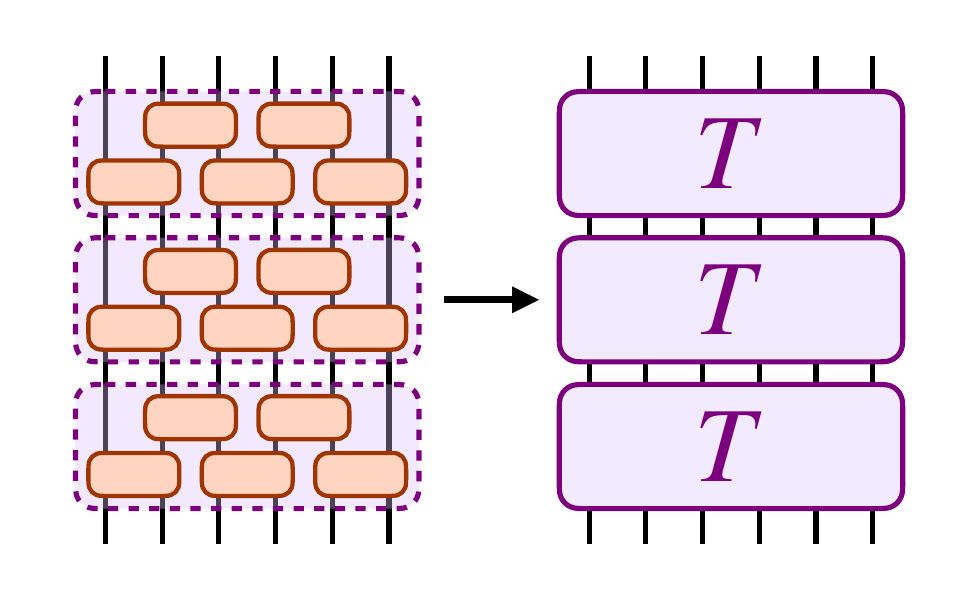}};
        \end{scope}
    \end{tikzpicture}
    \vspace*{-1.0cm}
    \caption{Breaking up \(\tmoment_\rqc\) for a random quantum circuit (in this case a 1D brickwork) into $\np=3$ copies of the transfer matrix $T$.}
    \label{fig:transfer_matrix}
\end{figure}

\section{Transfer matrix and the spectral gap} \label{section:finding_gap}
We now specialize to the case of \(\ell\)-layer periodic architectures. In this case we can define a transfer matrix \(T\) by contracting together the projectors $G_i, i \in \left\{1...\frac{\ell N}{2}\right\}$ of the moment operator $\tmoment_\rqc$ corresponding to a single period of the architecture, as shown in Figure~\ref{fig:transfer_matrix}. We will show that the approximate \(t\)-design time is controlled by the singular value spectrum of \(T\). 

If there are $k$ periods of the architecture, the moment operator corresponds to the \(k\)\textsuperscript{th} power of the transfer matrix, so by Equation~\ref{eq:frame_potential} the frame potential is
\begin{gather}
    \mathcal{F}_{\varepsilon} = \tr(T^{\dagger k} T^k) = ||T^k||^2_F
\end{gather}

\begin{theorem}
\label{thm:framepotential}
Consider a connected periodic architecture on \(N\) sites. After \(k\) periods, the frame potential is at most
    \begin{gather}
        \mathcal{F}_{\varepsilon} \leq \mathcal{F}_{\text{Haar}} + q^{2Nt}s_*^{2k}
    \end{gather}
    where $s_*$ is the largest non-unit singular value of the transfer matrix $T$.
\end{theorem}
\begin{proof}
Starting from
\[\mathcal{F}_{\varepsilon} = ||T^k||_F^2\]
we can use Theorem \ref{thm:periodic_normbound}, proven in Appendix \ref{app:ssv_bound}, to see
\[||T^k||_F^2 \leq m_1 + (d^2-m_1-m_0)s_*^{2k}\addtag\]
where \(d^2\) is the dimension of \(T\), \(m_1\) is the dimension of its unit eigenspace, and \(m_0\) is the dimension of its zero eigenspace. As long as the circuit is connected, by Lemma~\ref{lemma:global_is_unique}, it shares the same unit eigenstates as the Haar distribution, so \(m_1 = \mathcal{F}_\text{Haar}^{(q^N)}\). The first layer consists of \(2\)-site gates, each of which has only \(\mathcal{F}_\text{Haar}^{(q^2)}\) nonzero eigenvalues. A lower bound on \(m_0\) is thus \(d^2 -\left(\mathcal{F}_\text{Haar}^{(q^2)}\right)^N\). The total dimension for \(T\) is \(d^2 = q^{2Nt}\), so we find
\begin{gather}
    \label{eq:framebound_from_haars}
    \mathcal{F}_{\varepsilon} \leq  \mathcal{F}_\text{Haar}^{(q^N)} + \left(\left(\mathcal{F}_\text{Haar}^{(q^2)}\right)^N - \mathcal{F}_\text{Haar}^{(q^N)}\right)s_*^{2k}
\end{gather}
Of course, \(\mathcal{F}_\text{Haar}\) is just the dimension of the unit eigenspace, which is at most \(t!\) (since it is spanned by the uniform permutation states) and also at most \(q^{2Nt}\) (since that is the total dimension of \(T\)). In addition, it is known to be equal\cite{Hunter-Jones2019} to \(t!\) for \(t \leq d\). 
From Equation \ref{eq:define_frame_potential} it is clear that the frame potential is a monotonically increasing function of \(t\). Together these imply
\begin{gather}
    \min(t!,d!) \leq \mathcal{F}_\text{Haar}^{(d)} \leq \min\left(t!,d^{2t}\right)
\end{gather}
so that
\begin{gather*}
    \left(\mathcal{F}_\text{Haar}^{(q^2)}\right)^N - \mathcal{F}_\text{Haar}^{(q^N)} 
\leq \min\left(t!,q^{2t}\right)^N - \min\left(t!,\left(q^{N}\right)!\right)
\end{gather*}
The remaining algebra will be easier if we further relax this to
\begin{gather}
    \left(\mathcal{F}_\text{Haar}^{\left(q^2\right)}\right)^N - \mathcal{F}_\text{Haar}^{\left(q^N\right)} 
    \leq q^{2Nt}
\end{gather}
Without this relaxation, all of our bounds change by
\[2Nt \log q \rightarrow \log\left(\left(\mathcal{F}_\text{Haar}^{\left(q^2\right)}\right)^N - \mathcal{F}_\text{Haar}^{\left(q^N\right)}\right)\]
This latter is tighter for small \(t\), but the former is more convenient to interpret, so we prefer the simplified form. Equation \ref{eq:framebound_from_haars} then becomes
\begin{gather}
    \mathcal{F}_{\varepsilon} \leq  \mathcal{F}_{\text{Haar}}  + q^{2Nt}s_*^{2k}
\end{gather}
\end{proof}

Substituting Theorem \ref{thm:framepotential} into Equation \ref{eq:hunter-jones} gives a bound on  the rate at which a random circuit architecture approaches a $t$-design in terms of the subleading singular value (SSV) of the transfer matrix:
\begin{gather}
    \label{eq:diamondbound}
    ||\Phi_\text{RQC} - \Phi_\text{Haar}||_\diamond \leq q^{2Nt} s_*^k
\end{gather}
It follows that the number of periods required to push the diamond-norm error below $\epsilon$ can be upper-bounded in terms of \(s_*\) by
\begin{gather}
    \label{eq:criticaldepth}
    k_* = \frac{2Nt \log q + \log\frac{1}{\epsilon}}{\log\frac{1}{s_*}},
\end{gather}
which is already part of Theorem~\ref{thm:periodic_complete}. It remains only to bound \(s_*\), i.e. the spectral gap of $T$.

\section{Bounding the spectral gap} \label{section:layer_gap_bound}
In order to apply Theorem \ref{thm:framepotential}, we must bound the largest non-unit singular value \(T\). We will first show that this \(s_*\) is related to the geometry of the unit eigenspaces of each layer. Later we will study the relationship between this geometry and the architecture of the circuit to derive our final bound.

\(T\) can be interpreted as a product of orthogonal projection operators. Because each projector is norm-nonincreasing, any subspace whose norm is decreased by some large amount by the first few projectors cannot contain a large singular value. If \(P_i...P_1\) does not have a large singular value, then any large singular value of \(P_{i+1}...P_1\) must arise from vectors that were nearly unit eigenvectors of \(P_{i}...P_1\), while still being orthogonal to the unit eigenvectors of \(P_{i+1}...P_1\). This allows us to construct an inductive bound based on the relative geometry of the subspaces to which the \(P_i\) project.
\begin{theorem}
\label{thm:singvals1}
Consider some set of subspaces \(X_i, i \in \{1...n\}\) of a Hilbert space.
Let \(P_i\) be the orthogonal projector on to \(X_i\) and \(Q_i\) the orthogonal projector on to \(\bigcap_{j=1}^i X_j\). 
Define \(T = P_n...P_2 P_1\) and let \(s_*\) be the largest non-unit singular value of \(T_n\). Then we have the bound
\begin{gather}
    s_*^2 \leq 1 - \prod_{i=2}^n (1 - \mathscr{s}_i^2) \label{eq:singvals1}
\end{gather}
where \(\mathscr{s}_i\) is the largest non-unit singular value of \(P_i Q_{i-1}\).
\end{theorem}
The proof of this theorem is given in Appendix~\ref{app:ssv_bound}. In our case, the \(P_i\) will be the layers of the transfer matrix. We will call the \(\mathscr{s}_i\) the \textit{layer-restricted singular values}.  Our first goal is to characterize the \(Q_i\), which we term the \textit{intermediate unit eigenspace projectors}.

\subsection{Cluster-merging picture}
To bound the layer-restricted singular values, we first must understand the intermediate unit eigenspace to which \(Q_i\) projects. From Theorem~\ref{thm:projector}, we know that the unit eigenstates of a single gate are the uniform permutation states $\ket{\sigma}^{\otimes m}$. Moreover, each gate is norm non-increasing, so a state whose norm is reduced by any individual gate must not be a unit eigenstate. This leads to
\begin{lemma}
    \label{lm:cluster}
    Let \(G_i, i \in \{1...m\}\) be some gates within the transfer matrix which form a connected network on some $n$ sites. The unit eigenspace of the contracted network \(M = \prod_{i=1}^m G_i\) is spanned by the uniform permutation states $\{|\sigma\rangle^{\otimes n}, \sigma \in S_t\}$ 
\end{lemma}
\begin{proof}
We may interpret \(M\) as the transfer matrix of a (possibly incomplete) random quantum circuit architecture with the same layout as the \(G_i\). Since the \(G_i\) are connected, the corresponding quantum circuits are also connected. We may then apply Lemma \ref{lemma:global_is_unique}.
\end{proof}

\begin{figure}
    \centering
    \begin{tikzpicture}
        \begin{scope}
            \node[anchor=north west,inner sep=0] (image_a) at (0,0)
            {\includegraphics[width=0.99\columnwidth]{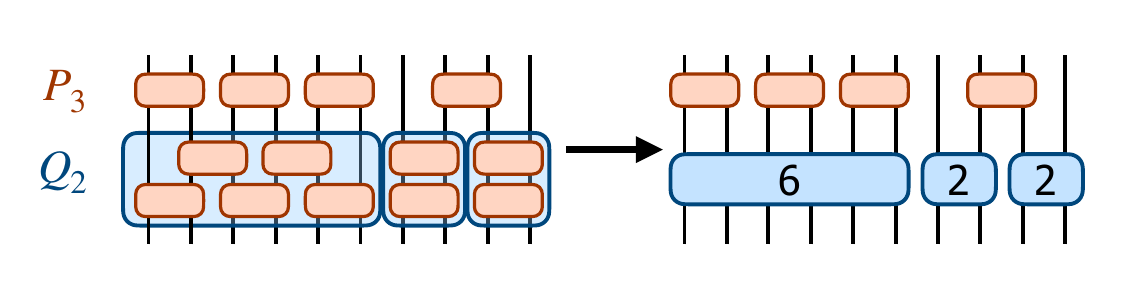}};
            \node [anchor=north west] (note) at (-0.15,-0.1) {\small{\textbf{a)}}};
        \end{scope}
    \end{tikzpicture}
    \vspace*{-0.5cm}
    \begin{tikzpicture}
        \begin{scope}
            \node[anchor=north west,inner sep=0] (image_a) at (0,-0.7)
            {\includegraphics[width=0.7\columnwidth]{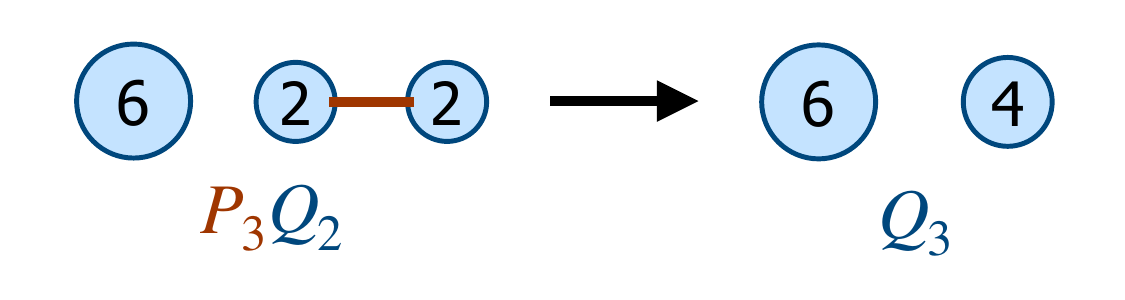}};
            \node [anchor=north west] (note) at (-0.15,-0.7) {\small{\textbf{b)}}};
        \end{scope}
    \end{tikzpicture}
    \caption{The cluster-merging picture for calculating the 3\textsuperscript{rd} layer-restricted subleading singular value $\mathscr{s}_3$, i.e. the subleading singular value of $P_3 Q_2$. (a): The first two layers are collected into clusters \textit{(blue)} of sites that are connected by the first two layers of gates. The intermediate unit eigenspace projector $Q_2$ is then a tensor product of projectors on each cluster. (b): (\textit{Left}) Graphical representation of $P_3 Q_2$ in the cluster merging picture. Each cluster in $Q_2$ becomes a node labeled by the number of sites it contains. Gates in $P_3$ that join distinct clusters become edges. Gates in \(P_3\) that join two sites within the same cluster are ignored. (\textit{Right}) Each connected component of this graph is merged into a single cluster of the next intermediate unit eigenspace projector $Q_3$ .}
    \label{fig:cluster_merging}
\end{figure}

This leads to the \textit{cluster-merging picture} (Figure~\ref{fig:cluster_merging}). Take some sequence of complete layers $P_1 ... P_l$. These layers may connect all the sites, or they may connect only certain subsets of the sites. Call each subset of sites which is connected by $P_1 ... P_l$ a \textit{cluster}. Recall that the intermediate unit eigenspace projector \(Q_l\) is defined to be the orthogonal projector on to the unit eigenspace of the product \(P_l ... P_1\). 

By Lemma~\ref{lm:cluster}, the space to which \(Q_l\) projects is the tensor product of the spans of the uniform permutation states on each cluster. As we include more layers by increasing \(l\), the clusters merge and the unit eigenspace to which \(Q_l\) projects shrinks. Eventually the whole circuit has been connected into a single cluster, at which time the unit eigenspace is the span of the globally uniform permutation states $|\sigma\rangle^{\otimes N}$. 

Our goal is to compute \(\mathscr{s}_l\), which is the largest non-unit singular value of \(P_{l}Q_{l-1}\). 
This can be accomplished by constructing the \textit{cluster-merging graph} which uniquely determines \(\mathscr{s}_l\). To build the graph, let each cluster of sites connected by \(P_{l-1}...P_1\) be a node, with weight equal to the number of sites in the cluster. Each gate of \(P_l\) that joins two distinct clusters is mapped to an edge joining the corresponding nodes. Gates of \(P_l\) that join two sites within the same cluster do not influence \(\mathscr{s}_l\) and can be ignored. The layer-restricted singular values \(\mathscr{s}_l\) depend only on the graph topology and the node weights, not on any other details of the architecture.

\section{Reduction of each layer to 1D brickwork loops} \label{section:cluster_splitting}
Our next goal is to obtain an architecture-independent upper bound on the \(\mathscr{s}_i\). We begin by identifying a set of rules for rearranging a cluster-merging graph into a standardized form without decreasing \(\mathscr{s}_i\). 

\subsubsection{Structure of the graph}
\begin{lemma}
    For a complete \(\ell\)-layer periodic architecture, the cluster-merging graph for each layer above the first has nodes of even weight and even degree.
\end{lemma}
\begin{proof}
We first show that the clusters are of even size. The first layer creates clusters of size \(2\). Subsequent layers create clusters by merging these size-\(2\) clusters together, so later cluster sizes are also even. 

We next show that each cluster has an even number of external connections.
Suppose a cluster of size \(m\) has \(n_i\) internal gates and \(n_e\) external gates. Since each site is acted on by exactly one two-site gate, we have \(2n_i + n_e = m\), and so \(n_e\) is even. The nodes are thus of even degree.
\end{proof}

\subsubsection{Cluster-merging bound}
Our goal is to bound \(\mathscr{s}_i\) for generic cluster-merging graphs in terms of cluster-merging graphs of some standard form. First consider the unit eigenspace.
\begin{lemma}
\label{lemma:cluster_unit_space}
The unit eigenspace of a connected cluster-merging graph is spanned by the globally uniform permutation states.
\end{lemma}
This is just a special case of Lemma \ref{lm:cluster}.
We can now show a lower bound:
\begin{lemma}
\label{lemma:merge}
Let \(A\),\(B\) be two nodes of a connected cluster-merging graph. If we merge \(A\) and \(B\) into a single node \(AB\), the subleading singular value of the graph does not increase.
\end{lemma}
\begin{proof}
Let \(\mathscr{s}, \mathscr{s}'\) be the subleading singular value of the old and new graphs, respectively. Let \(X,X'\) be the subspaces projected to by the nodes of the old and new graphs. Let \(P\) be the projector corresponding to the edges, which are the same for both graphs.

By Lemma \ref{lemma:cluster_unit_space}, the unit eigenspaces of \(PQ\) and \(PQ'\) are the same. Call this subspace \(Z\). We may write
\[\mathscr{s} = \max_{\left\{v \in  X| v \perp Z \right\}} \frac{||P v||}{||v||}\addtag\]
and
\[\mathscr{s}' = \max_{\left\{v \in  X'| v \perp Z \right\}} \frac{||P v||}{||v||}\addtag\]
\(X\) is spanned the cluster-uniform states over the old nodes, while \(X'\) is spanned by the cluster-uniform states over the new nodes. A cluster-uniform permutation state on the new node \(AB\) is of the form \(\ket{\sigma}^{\otimes (|A| + |B|)}\), which is also cluster-uniform on \(A\) and \(B\) individually. In other words, \(X' \subseteq X\). Since \(\mathscr{s}'\) is the maximum of the same function over a smaller space, we obtain \(\mathscr{s}' \leq \mathscr{s}\).
\end{proof}
If we run our lower bound in reverse, we are led to the following two rewriting rules that give upper bounds on subleading singular values of cluster-merging graphs:
\begin{lemma}
\label{lemma:split}
If we split a node and the two sides remain part of the same connected graph, the subleading singular value does not decrease.
\end{lemma}

\begin{lemma}
\label{lemma:splitlink}
If we split a node and add a new link to keep the graph connected, the subleading singular value does not decrease.
\end{lemma}
\begin{proof}
These both correspond to applying \ref{lemma:merge} backwards. 
\end{proof}
Note that the number of links connected to a cluster cannot exceed the number of sites it contains, so we can only apply the latter lemma to clusters with at least two unoccupied sites. We are now ready to prove the main result of this section.

\begin{figure*}
    \centering
    \begin{tikzpicture}
        \begin{scope}
            \node[anchor=north west,inner sep=0] (image_a) at (0,0)
            {\includegraphics[width=1.3\columnwidth]{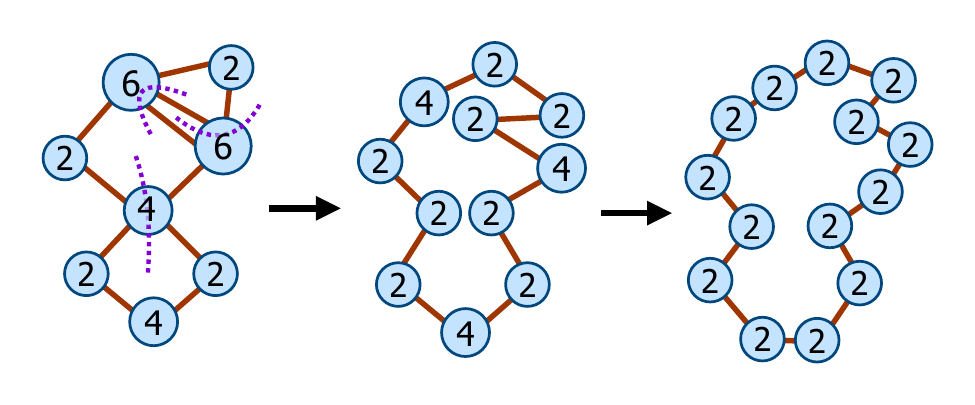}};
        \end{scope}
    \end{tikzpicture}
    \vspace*{-0.6cm}
    \caption{Algorithm for bounding a cluster-merging graph by a 1D brickwork structure. We first use Lemma \ref{lemma:split} to split clusters (without joining the halves with an edge) along an Eulerian circuit in our graph.We then use Lemma \ref{lemma:splitlink} to break up larger clusters (joining the halves with an edge) until all clusters in the loop have size 2. Each step in this process does not decrease the subleading singular value of the graph.}
    \label{fig:cluster_splitting}
\end{figure*}

\begin{theorem}
\label{thm:clusterbound}
Consider any cluster-merging graph on \(n\) sites with SSV \(\mathscr{s}\). Let \(s_\text{1D}(m)\) be the SSV for a 1D brickwork loop on \(m\) sites. We have
\[\mathscr{s} \leq \max_{m \leq n} s_\text{1D}(m)\addtag\]
\end{theorem}
\begin{proof}
Our goal is to apply the graph rewriting rules repeatedly to upper-bound a cluster in a standard form. First suppose that the graph is connected. Since nodes are of even degree, there exists an Eulerian circuit through the cluster-merging graph. We can split nodes along this Eulerian circuit using Lemma \ref{lemma:split} until our graph structure is a single loop, again with nodes of even degree (Figure~\ref{fig:cluster_splitting}, first process). We can then apply Lemma \ref{lemma:splitlink} repeatedly within each node to split each node into many nodes, each of size \(2\) (Figure~\ref{fig:cluster_splitting}, second process). When we are done, we have exactly a 1D brickwork loop on \(m\) sites. These transformations cannot decrease the SSV, so the brickwork SSV is an upper bound on the original SSV. 

Now consider any disconnected graph. We can apply the argument above to bound each connected component by a loop graph. The corresponding operator is a tensor product of loop operators, and the largest singular value of each loop operator is \(1\). The  subleading singular value of the whole operator is then just the largest subleading singular value of any of the connected components. Each connected component is a brickwork loop of size at most \(m \leq n\). 
\end{proof}

\section{Spectral gap of 1D Brickwork Loops} \label{section:1d_solution}
In order to extract a useful result from Theorem \ref{thm:clusterbound}, we need a bound on the spectral gap of the 1D brickwork architecture. In Section \ref{section:cluster_splitting} we showed that the $t$-design depth is controlled by the spectral gap of the transfer matrix. We now reverse that argument in order to obtain a bound on the spectral gap in terms of the $t$-design depth. This will allow us to convert any result on 1D brickwork $t$-design depths into bounds on the spectral gap.

\begin{theorem}
\label{thm:1deigfromtdesign}
Suppose that the 1D brickwork architecture on \(N\) sites (with either open or periodic boundary conditions) forms an \(\epsilon\)-approximate \(t\)-design after
\begin{gather}
    k_1 = C(N,q,t) \log \frac{1}{\epsilon} + o_\epsilon \left(\log \frac{1}{\epsilon}\right) \label{eq:kstar}
\end{gather}
periods (corresponding to depth \(2k_1\)), for some function \(C(N,q,t)\).
Then the largest non-unit singular value \(s_\text{1D}\) of the corresponding transfer matrix \(T_t\) is bounded by
\[s_\text{1D}  \leq e^{-\frac{1}{2C(q,t)}}\addtag\] 
where $C(q,t) = \sup_{N} C(N,q,t)$
\end{theorem}
\begin{proof}
Define
\[\Delta(k,t) = \Phi_{\text{RQC},k,t} - \Phi_\text{Haar}\]
From the definition of the diamond norm 
\[||\Delta||_\diamond \geq ||\Delta||_{1 \rightarrow 1} \geq \frac{||\Delta(\rho)||_1}{||\rho||_1}\addtag\]
for any operators \(\rho\). Let \(\lambda\) be the largest eigenvalue of the \(\Delta\) and choose \(\rho\) to be the corresponding eigenvector, so that \(\Delta(\rho) = \lambda \rho\). Then we obtain
\begin{equation}
    \label{eq:diamond_lower_bound}
    ||\Delta||_\diamond \geq |\lambda|
\end{equation}
Note that the leading eigenvalue of \(\Delta\) is exactly the subleading eigenvalue of \(\Phi_\text{RQC}\), since the unit eigenspace of \(\Phi_\text{RQC}\) is exactly cancelled by \(\Phi_\text{Haar}\). Furthermore, \(\tmoment_\rqc = T_t^k\), so we have
\[\lambda = \lambda_*^k\addtag\]
where \(\lambda_*\) is the subleading eigenvalue of \(T\). If we choose \(k = k_1\) so that
\[||\Delta(k_1,t)||_\diamond < \epsilon\]
we obtain
\[|\lambda_*| < \epsilon^\frac{1}{k_1}\addtag\]
Inverting equation (12) gives
\[\log \frac{1}{\epsilon} = \frac{k_1}{C(N,q,t)} + o_{k_*}(k_*)\addtag\]
or
\[\epsilon = e^{-\frac{k_1}{C(N,q,t)} + o_{k_1}(k_1)}\addtag\]

We can now take the limit of small \(\epsilon\) or equivalently large \(k_1\) to obtain
\[\epsilon = \left(e^{-\frac{1}{C(N,q,t)}}\right)^{k_1}\addtag\]
which implies \(|\lambda_*| \leq e^{-\frac{1}{C(q,t)}}\). Theorem \ref{thm:2layer_eigvsing} of Appendix \ref{app:ssv_bound} tells us \(s_\text{1D} = \sqrt{\lambda_*}\), so we find
\(s_\text{1D} \leq e^{-\frac{1}{2C(q,t)}}\)
\end{proof}

\begin{theorem} 
\label{theorem:open_brickwork}
Let \(s_\text{1D,open}(N)\) be the subleading singular value of the $N$-site 1D brickwork with open boundary conditions, and let \(s_\text{1D,periodic}(N)\) be the same for the periodic-boundary-conditions case. Then \(s_\text{1D,periodic}(N) \leq s_\text{1D,open}(N)\)
\end{theorem}
This result allows us to also use bounds derived for open brickworks in Theorem \ref{thm:clusterbound}. It follows directly from the following more general rewriting rule for cluster-merging graphs.
\begin{lemma}
\label{lemma:add_new_link}
Consider a connected cluster-merging graph. If we add a new edge to the graph, the subleading singular value \(\mathscr{s}\) does not increase.
\end{lemma}
\begin{proof}
Let \(P\) and \(Q\) be the edges and nodes of the original graph, and let \(R\) be the new link. Let \(X\) be the unit eigenspace of \(Q\) and \(Z\) the unit eigenspace of \(PQ\). Since the graph was already connected without \(R\), the unit eigenspace of \(R\) includes \(Z\), so \(Z\) is also the unit eigenspace of the new graph \(RPQ\). The original singular value is 
\[\mathscr{s} = \max_{\left\{v \in X | v \perp Z \right\}} \frac{||P v||}{||v||}\addtag\]
while the new singular value is
\[\mathscr{s}' = \max_{\left\{v \in X | v \perp Z \right\}} \frac{||RP v||}{||v||}\addtag\]
Since \(R\) is an orthogonal projector, \(||RPv|| \leq ||Pv||\), so \(\mathscr{s}' \leq \mathscr{s}\).
\end{proof}
Given an open-boundary-condition transfer matrix, we can add a link to obtain the transfer matrix of a periodic-boundary-condition brickwork on the same number of sites. Adding a link can only decrease \(\mathscr{s}\), which completes the proof of Theorem \ref{theorem:open_brickwork}.

\section{Approximate \(t\)-design depths} \label{section:summary}
We are now ready to prove Theorem \ref{thm:periodic_complete}, which addresses the case of complete periodic architectures. We will then proceed to extend our results to incomplete and aperiodic architectures. Figure \ref{fig:circuit_classification} illustrates the relevant categories.

\subsection{Complete periodic architectures}
\label{subsection:periodic_proof}
Suppose that the \(N\)-site 1D brickwork architecture forms an \(\epsilon\)-approximate \(t\)-design in depth at most
\begin{gather}
2C(N,q,t) \log \frac{1}{\epsilon} + o_\epsilon \left(\log \frac{1}{\epsilon}\right)
\end{gather}
Define \(C(q,t) = \sup_N C(N,q,t)\). 
By Theorem \ref{thm:1deigfromtdesign} the subleading singular value of \(T_t\) for a brickwork loop is upper-bounded by
\begin{gather}
s_\text{1D} \leq e^{-\frac{1}{2C(q,t)}}
\end{gather}
From Theorem \ref{thm:clusterbound} it follows that the SSV for any cluster-merging graph is bounded by the same value. Substituting into Theorem \ref{thm:singvals1}, we see that the subleading singular value for a complete \(\ell\)-layer block is bounded by
\begin{gather}
    \label{eq:final_singvalbound}
    s_* \leq 1 - \left(1 - e^{-\frac{1}{2C(q,t)}}\right)^{\ell - 1}
\end{gather}
We can substitute into Equation \ref{eq:criticaldepth} to see that the critical period count is then bounded by
\begin{gather}
    \label{eq:final_periodic_bound}
    k_* = \frac{2Nt \log 2 q + \log \frac{1}{\epsilon}}{\log \left[\left(1 - \left(1 - e^{-\frac{1}{2C(q,t)}}\right)^{\ell - 1}\right)^{-1}\right]}
\end{gather}
This completes the proof of Theorem \ref{thm:periodic_complete}.

\subsection{Incomplete layers}
\label{subsection:incomplete_proof}

For complete architectures, we have
\begin{gather}
    s_*^2 \leq 1 - (1 - s_\text{1D}(q))^{\ell-1}
\end{gather}
by reducing each cluster-merging graph to an Eulerian cycle. If the circuit is incomplete, the Eulerian cycle is not guaranteed to exist. Instead, an incomplete circuit with local dimension \(q\) has subleading singular value bounded by that of a complete circuit with local dimension \(\sqrt{q}\).

We may, without loss of generality, insert single-site gates so that each
site is acted upon by exactly one gate per layer. These single-site gates can be absorbed into a two-site gate
either above or below without changing the Haar
measure over that two-site gate. This takes an incomplete circuit to a
``complete'' circuit with a mixture of one-site and two-site gates. 

Now suppose that \(\sqrt{q}\) is an integer. We can then split each site of dimension \(q\) into a pair of sites of dimension \(\sqrt{q}\). This gives us a complete circuit containing a mixture of \(2\)- and \(4\)-site gates. Suppose we split each \(4\)-site gate lengthwise into a pair of \(2\)-site gates. It seems intuitive that this operation should decrease the rate of scrambling, not increase it. And the resulting circuit is a complete circuit of \(2\)-site gates with local dimension \(\sqrt{q}\), so its cluster-merging graphs have Eulerian cycles. In Appendix~\ref{app:irrational_q}, we formalize this intuition and show that the same idea can be applied even in the case where \(\sqrt{q}\) is not an integer. The site-splitting process is illustrated in Fig. \ref{fig:site_splitting}. 

This allows us to reduce a bound on incomplete layers to a 1D brickwork bound on $\sqrt{q}$, specifically
\begin{gather}
    s_*^2 \leq 1 - (1 - s_\text{1D}(\sqrt{q}))^{\ell-1}
\end{gather}
for incomplete architectures. 

We may use any proof of \(C(q,t)\) if $\sqrt{q}$ is an integer. For non-integer $\sqrt{q}$, however,
the definition of \(s_\text{1D}\) is trickier. It is not necessarily true that generic bounds on $C(q,t)$ which have been derived for integer \(q \geq 2\) can be analytically continued to $C(\sqrt{q},t)$. Many of the proofs in the literature can easily be extended to non-integer \(q > 1\) (see Section \ref{subsection:cvals}). However, not all of these proofs can be easily extended, e.g. ref.~\onlinecite{Haferkamp2022} which applies only to \(q = 2\). The bound of ref.~\onlinecite{Brandao2016} can be extended, but the scaling goes from \(\sim t^{10}\) to \(\sim t^{16}\). And there may be future improved strategies for bounding 1D brickwork that work only for integer \(q\). We thus also give the following bound, which gives us a reduction in terms of \(C(q,t)\) instead of \(C(\sqrt{q},t)\). 

\begin{theorem}\label{thm:periodic_incomplete_nointdefined}
For an architecture with blocks of $\ell$ layers,
\begin{gather}
    k_* = \big(4\overline{C}(q,t)\big)^{x(N) \ell-1} \left(2Nt \log q + \log \frac{1}{\epsilon}\right)
\end{gather}
where the expansion coefficient
\begin{gather}
    x(N) = 8\lceil \log_2 \lfloor \log_2(N+1)\rfloor \rceil + 2
\end{gather}
\end{theorem}
The proof is found in Appendix~\ref{app:treebound}.

\subsection{Aperiodic architectures}
\label{subsection:aperiodic_proof}
So far we have restricted ourselves to periodic architectures to simplify the exposition. However, our results generalize quite directly to aperiodic architectures.

\begin{theorem}
\label{thm:aperiodic}
Define \(C(q,t)\) as in Theorem \ref{thm:periodic_complete}. For a not-necessarily-periodic \(L\)-layer circuit \(U_C\), choose a decomposition
\begin{gather}
    \label{eq:circuit_split}
    U_{C} = V_{k+1} U_{k} V_{k} U_{k-1} V_{k-1}  ...U_{1} V_1
\end{gather}
where each \(U_{i}\) is a connected $\ell_i$-block and each \(V_i\) is some contiguous block of layers which may not connect all the sites. The architecture is an \(\epsilon\)-approximate \(t\)-design if
\begin{gather}
k \geq \frac{2Nt \log q + \log \frac{1}{\epsilon}}{\log\frac{1}{s_*}}
\end{gather}
where we have defined the effective averaged singular value
\[s_* = 1 - \left(1 - e^{-\frac{1}{2C(\sqrt{q},t)}}\right)^{\bar{\ell}-1}\]
in terms of the mean block size
\[\bar{\ell} = \frac{1}{k}\sum_{i=1}^k \ell_i\]
\end{theorem} 
If we also require that the layers be complete, we can replace \(\sqrt{q}\) with \(q\). Notice that \(k\) counts the total number of times the circuit is connected, while \(\bar{\ell}\) is the size of the typical connected block. We see that the \((\epsilon,t)\)-design depth is controlled by the frequency and size of the connected blocks. However, for a aperiodic architecture \(\bar{\ell}\) might depend on depth, so this does not give us any explicit expression for the critical depth. 

\begin{proof}
A decomposition of the circuit into blocks of layers induces a decomposition of the associated \(t\)-fold channel as
\[\hat{\Phi}_\text{RQC} = R_k T_{k} R_{k-1} T_{k-1}...R_1 T_1\]
where the $R_i$ are some set of norm-nonincreasing transfer matrices. Let \(s_*^{(i)}\) be the subleading singular value of each \(T_i\). Equation \ref{eq:final_singvalbound} still applies for each \(s_*^{(i)}\), with \(\ell\) replaced by \(\ell_i\).\footnote{We can further tighten the bound for certain architectures by observing that \(\ell_i\) only needs to count layers that merge clusters. Layers consisting of only internal edges within already-merged clusters need not be counted.} 
Theorem \ref{thm:normbound_with_extras} of Appendix \ref{app:ssv_bound} shows that the \(R_i\) are essentially irrelevant. Equation \ref{eq:diamondbound} now becomes
\begin{gather}
    ||\Phi_\text{RQC} - \Phi_\text{Haar}||_\diamond \leq q^{2Nt} \prod_{i=1}^k s_*^{(i)}
\end{gather}
The condition to obtain an \(\epsilon\)-approximate \(t\)-design is then
\begin{gather}
    \label{eq:generic_errorbound}
    \epsilon \leq q^{2Nt} \prod_{i=1}^k \left(1 - \left(1 - s_\text{1d}\right)^{\ell_i - 1}\right)
\end{gather}
We may rearrange this as
\begin{gather}
    2Nt \log q + \log \frac{1}{\epsilon} \geq -\sum_{i=1}^k \log\left(1 - \left(1 - s_\text{1d}\right)^{\ell_i - 1}\right)
\end{gather}
Furthermore, \(-\log(1 - c^x)\) is convex, so by Jensen's inequality
\begin{gather}
    \label{eq:other_generic_errorbound}
    2Nt \log q + \log \frac{1}{\epsilon} \geq -k\log\left(1 - \left(1 - s_\text{1d}\right)^{\bar{\ell} - 1}\right)
\end{gather}
\end{proof}

This formula implies Theorem \ref{thm:periodic_incomplete} as a special case. Furthermore, it is also simple for \textit{regularly-connected} architectures for which all of the \(\ell_i\) are equal. 

There is still a question of the choice of decomposition in Equation \ref{eq:circuit_split}. This decomposition is not unique; different choices of decomposition will give different bounds. Note in particular that \(V_i\) may be empty, which corresponds to the identity.  

Furthermore, the optimal decomposition may depend on \(q\) and \(t\). For example, consider an architecture consisting of alternating connected \(2\)-blocks and connected \(4\)-blocks. If we count all the blocks, then \(\bar{\ell} = 3\) and the depth is \(d = 3k\). But if we count only the \(2\)-blocks and lump the \(4\)-blocks into the \(V_i\), we obtain \(\bar{\ell} = 2\) and \(d = 6k\). The former is better for small \(C\) and the latter for large \(C\), with a crossover point \(C\left(\sqrt{q},t\right) \approx 1.157\).

\section{Further extensions} \label{section:further_extensions}
\subsection{\(\ell\)-independent bound}
\label{subsection:connection_conjecture}
It is interesting to note that our Theorem~\ref{thm:periodic_complete} gives a bound which loosens as the period \(\ell\) increases. It seems likely that the approximate \(t\)-design depth actually decreases\cite{Harrow2009}  with \(\ell\) for certain well-connected structures, such as higher-dimensional brickworks. But the scaling of our bound, which is determined by \textit{worst-case architectures}, suggests that there may exist strange ``tenuously connected'' architectures at larger \(\ell\). 

It seems intuitively clear that the open-boundary-condition brickwork architecture is in some sense the most ``spread out'' arrangement of gates possible. Any other connected architecture must involve more total gates, and graph distances between sites must be shorter. Indeed, optimization of the subleading singular value of small circuits using simulated annealing (see Appendix \ref{app:connection_numerics}) have failed to find any of these ``tenuously connected'' architectures: All subleading singular values are bounded by that of the open-boundary-condition brickwork. This motivates the following conjecture:
\begin{conjecture}
    \label{conjecture:connection_count}
    Every connected architecture on \(N\) sites has subleading singular value \(s_* \leq s_{\text{1D,open}}(N,q,t)\)
\end{conjecture}
The immediate consequence of this conjecture is that there is a universal \(t\)-design \textit{connection count} which does not depend on the circuit architecture even via \(\ell\). In particular, any circuit architecture that can be divided into 
\begin{equation}
    \label{eq:connection_count}
    k_* = \frac{2Nt \log q + \log \frac{1}{\epsilon}}{\log \frac{1}{s_{\text{1D,open}}}}
\end{equation}
connected blocks forms an \((\epsilon,t)\)-design. 

To obtain Theorem \ref{thm:best_guess}, we also follow ref.~\onlinecite{Hunter-Jones2019} in making the following guess:
\begin{conjecture}
    \label{conjecture:s_exact_open}
    \begin{equation}
        \label{eq:C_exact_open}
        s_{\text{1D,open}}(q,t) = \frac{2q}{q^2 + 1}
    \end{equation}
\end{conjecture}
Numerical evidence for this formula is given in Appendix \ref{app:brickwork_t_independence}. 

\subsection{Nondeterministic architectures}
Our theorems focus on deterministic architectures, in which the contents of the gates are random but their arrangement is fixed. Another interesting class of ensembles is nondeterministic architectures, in which the locations of the gates are also random\cite{Brandao2016}.  
In this case, the \(t\)-design property is obtained by averaging over both the spatial arrangement of the gates and their content. Bounds on nondeterministic and deterministic architectures can often be related to each other by the union bound or detectability lemma\cite{Aharonov2010, Haferkamp2022}. 

We can also show more directly that bounds for particular spatial structures imply bounds for averages over ensembles of structures.
We use the triangle inequality on \(||\Phi_\text{RQC} - \Phi_\text{Haar}||_\diamond\), where \(\Phi_\text{RQC}\) is drawn from some distribution \(\rho_\Phi\). Then
\begin{gather}
    ||\langle \Phi_\text{RQC} \rangle_{\rho_\Phi} - \Phi_\text{Haar}||_\diamond \leq \left \langle|| \Phi_\text{RQC} - \Phi_\text{Haar}||_\diamond\right\rangle_{\rho_\Phi}
\end{gather}
In particular Equation \ref{eq:generic_errorbound} with the right-hand side replaced with its average over the architecture. We can use Jensen's inequality again to see
\begin{gather}
    \log \left \langle|| \Phi_\text{RQC} - \Phi_\text{Haar}||_\diamond\right\rangle_{\rho_\Phi} \leq 
    \left \langle\log || \Phi_\text{RQC} - \Phi_\text{Haar}||_\diamond\right\rangle_{\rho_\Phi}
\end{gather}
so we can also apply Equation \ref{eq:other_generic_errorbound}  with the right-hand side replaced by its average. 

The connection count \(k\) and mean block size \(\bar{\ell}\) will both differ between realizations, and it is not clear that there is any general strategy for calculating their distribution. Presumably they are not independent, so we can't solve Equation \ref{eq:other_generic_errorbound} for \(\langle k \rangle\). An interesting open question is whether there exists any simple relationship between the averaged bound and the distribution from which the circuit structure is drawn.

A commonly-studied example is that of gates which are drawn sequentially from some a uniform distribution over some set of pairs of sites \cite{Brandao2016, Haferkamp2022, Oszmaniec2022}. In the worst case, such circuits may require \(N-1\) nontrivial layers to connect all the sites\footnote{The upper bound is \(N-1\) because we ignore layers with completely disconnected cluster-merging graphs.}. However, \textit{typical} instances probably require far fewer distinct layers, perhaps as few as \(O(1)\). Further work may wish to explore the relationship between bounds obtained by this strategy and known bounds for such architectures \cite{Oszmaniec2022}. 

If we assume Conjecture \ref{conjecture:connection_count}, the nondeterministic case becomes more tractable. We find that it forms an approximate $t$-design when
\begin{equation}
\langle k \rangle \geq \frac{2Nt \log q + \log \frac{1}{\epsilon}}{\log \frac{1}{s_{\text{1D,open}}}}
\end{equation}
Suppose we have \(n_g\) gates sampled i.i.d from the uniform distribution over the edges of some large connected graph. There is quite a bit more we can say about the relationship between \(n_g\) and \(\langle k \rangle \). The fully-connected graph, for example, undergoes a percolation phase transition and becomes connected after \(O(N\log N)\) gates, so \(n_g = O\left(N\log N\right)\langle k \rangle\).\cite{Erdos1959} This corresponds to an approximate $t$-design threshold size of $n_g = O(N^2 \log N)$. For $t=2$ this is worse than the known bound by a factor of $\log N$.\cite{Ambainis2007} From the coupon collector problem, we see that the linear graph also becomes connected after $O(N \log N)$ gates, so its threshold is the same.\cite{Erdos1961} Again a bound which is better by a factor of $\log N$ (and which does not rely on Conjecture \ref{conjecture:connection_count}) is already known.\cite{Brandao2016}

For a graph with \(E\) edges, we can again use the coupon collector problem to see that it must be connected after $O(E \log E)$ gates. Since $E \leq O(N^2)$, this suggests every graph gives an approximate $t$-design threshold size of at most $n_g = O(N^3 \log N)$. For graphs which admit Hamiltonian paths, this result is also weaker by $\log N$ than what was already known, but for graphs without Hamiltonian paths it would be a new result.\cite{Oszmaniec2022}

\subsection{Highly-connected architectures}
Previous work on random circuits\cite{Brandao2016, Hunter-Jones2019, Haferkamp2022, Harrow2009} has focused largely on brickwork architectures. Brickwork architectures are in some sense exceptionally well-connected, so they should be expected to converge to the Haar distribution relatively quickly. Indeed, ref.~\cite{Harrow2023} shows certain higher-dimensional brickwork circuits approach the Haar measure at a rate that increases with geometric dimension, which corresponds to increasing the period \(\ell\) of the architecture. Here we suggest an extension of our techniques to obtain tighter bounds for such special architectures. 

The first observation is that highly-connected cluster-merging graphs give small \(\mathscr{s}_i\). Second, certain families of $N$-site architectures form clusters of sizes which scale with $N$ after only a few layers. If we join two $m$-site clusters with $m$ gates, Appendix \ref{app:brickworks} shows that the subleading value is \(q^{-\Omega(m)}\) as \(m \rightarrow \infty\). In other words, any layers which only join clusters with an extensive number of edges do not contribute to our bound.

For the specific case of higher-dimensional brickwork circuits and \(t=2\), we can give a more explicit calculation. Consider a $D$ dimensional lattice of $N=L^D$ points at positions $\vec{x} \in \{1 ... L\}^D$, with $\hat{e}_i$ being the unit vector in direction $i$, and connect the lattice with the following sequence of layers: 
\begin{itemize}
    \item At layers $1 \leq i \leq D$, we apply gates on pairs $(\vec{x}, \vec{x}+\hat{e}_i)$ for each $\vec{x}$ such that $x_i$ is even (all addition of components is performed modulo $L$, to make this architecture periodic).
    \item At layers $D+1 \leq i \leq 2D$, we apply gates on pairs $(\vec{x}, \vec{x}+\hat{e}_{i-D})$ for each $\vec{x}$ such that $x_{i-D}$ is odd.
\end{itemize}
This is a higher dimensional generalization of the brickwork architecture that emphasizes accessing all dimensions as quickly as possible, instead of repeating the one-dimensional brickwork across multiple directions. 

The first \(D\) (``even'') layers of the \(D\)-dimensional brickwork form hypercube clusters of size \(2^D\). The \(D+1\)th layer, which is the first ``odd'' layer, then connects these hypercubes into rows of \(2^D N^\text{1/D}\) sites. The remaining odd layers are then highly connected, so in the limit of large \(N\) their contribution can be ignored. This argument, given in detail in Appendix \ref{app:brickworks}, shows that the \((\epsilon,t=2)\)-design depth of the \(D\)-dimensional brickwork is at most
\begin{gather}
    d_* = 2D \frac{4N \log q + \log \frac{1}{\epsilon}}{\log \left[\left(1 - \left(1 - \exp\left(-\frac{1}{2C(q,2)}\right)\right)^2\right)^{-1}\right]}
\end{gather}
It seems likely that this argument could be improved to give a bound that actually decreases with \(D\), since the dominant \(D+1\)\textsuperscript{th}-layer cluster-merging graph has a relatively simple structure. It also may be possible to extend such techniques to larger \(t\) and for other extensively-connected architectures. 

\section{Conclusion} \label{section:conclusion}
We show that bounds on approximate $t$-designs for 1D brickworks imply bounds for general architectures. This process not only gives us an immediate bound linear in $N$ for the $t$-th moments of all sufficiently well-connected architectures, but allows us to convert any improved 1D bounds into bounds on generic structures. We also show that our bounds can be extended by an averaging procedure to an implicit bound for nondeterministic architectures.

Any architecture consisting of connected $\ell$-blocks of $O(1)$ depth can be bounded this way. So our result implies that any sufficiently regularly-connected circuit ensemble approximates global information scrambling in at most linear depth, albeit with exponential dependence on the connection frequency.

We conjecture that this bound can be tightened to one which depends on the circuit architecture only via the \textit{connection count}. This suggests rapid scrambling is inescapable for any sufficiently well-connected architecture.

\textit{Note added:}  During the final stages of completing the present manuscript, we became aware of a related parallel independent  work with similar results on nondeterministic architectures which will appear in the same arXiv posting~\cite{Mittal2023}. 

\begin{acknowledgements}
We acknowledge useful conversations with Nicholas Hunter-Jones.  D.B. acknowledges Felix Leditzky for suggesting relevant literature. This material is based upon work supported by the U.S. Department of Energy, Office of Science, National Quantum Information Science Research Centers. B.F. and S.G. acknowledge support from AFOSR (FA9550-21-1-0008) and DOE QuantISED grant DE-SC0020360.
This material is based upon work partially
supported by the National Science Foundation under Grant CCF-2044923 (CAREER).

\end{acknowledgements}
\bibliographystyle{unsrt}
\bibliography{references2}

\begin{appendices}

\section{Properties of projector products} \label{app:ssv_bound}
In this section we prove some general results about products of orthogonal projection operators. Consider some set of subspaces \(X_i, i \in \{1...n\}\) of a Hilbert space.
Define \(P_i\) to be the orthogonal projector on to \(X_i\) and
\[T = P_n...P_2 P_1\]
In general we are interested in understanding the singular value spectrum of \(T\).

\subsection{Structure of the unit eigenspace}
\begin{lemma}
    The left unit eigenspace, right unit eigenspace, and unit singular value space of $T$ are all \(\bigcap_{i=1}^n X_i\). \label{lemma:leftright_eigenspace}
\end{lemma}
\begin{proof}
Let \(Y\) be the unit eigenspace of \(T\). A projector is norm-nonincreasing (i.e. \(||P_i||_\infty = 1\)). Furthermore, it acts as the identity on any vector whose norm it does not decrease. It follows that a unit eigenvector of \(T\) must be a unit eigenvector of each of the \(P_i\). It's easy to see that the converse also holds, so \(Y = \bigcap_{i=1}^n X_i\). 
This argument works the same from the left and the right, so the left and right unit eigenspaces are the same.

Now consider the singular value spaces of \(T\). These are the square roots of the eigenvalues of \(T^\dagger T\). Since the left and right eigenspaces are the same, they are in the unit eigenspace of \(T^\dagger T\). Since both \(T^\dagger\) and \(T\) are norm-nonincreasing, they must be the whole unit eigenspace. So the unit singular-value-space of \(T\) is exactly the unit eigenspace. 
\end{proof}

\begin{lemma}
    The unit eigenspace of $T$ is orthogonal to all other eigenstates of $T$, all other eigenstates of $T^{\dagger k} T^k$, and remains orthogonal no matter how many factors of $T$ or $T^\dagger$ are applied to the other eigenstate.\label{lemma:eig_is_ortho}
\end{lemma}

\begin{proof}
From Lemma \ref{lemma:leftright_eigenspace} we know that the left and right unit eigenspaces of \(T\) are the same. Let \(u\) be a unit eigenvector and \(v\) a right-eigenvector with eigenvalue \(\lambda < 1\). We can compute 
\[u^\dagger v = u^\dagger T v = \lambda u^\dagger v\addtag\] which can hold only if \(u^\dagger v = 0\). So \(Y\) is orthogonal to every non-unit eigenspace of \(T\). 

Since $T^{\dagger k} T^k$ is Hermitian, its eigenstates of different eigenvalues are automatically orthogonal, so $u^\dagger v = 0$ for any subunit eigenstates $v$ of $T^{\dagger k} T^k$.

These proofs still hold if we apply extra factors of $T$ or $T^\dagger$ to $v$, because we can freely absorb these extra factors into $u$.

\end{proof}

\subsection{Bound from layer-restricted subleading singular values}
\begin{theorem}
\label{thm:singvals}
Let \(Q_i\) to be the orthogonal projector on to \(\bigcap_{j=1}^i X_j\). Let \(\mathscr{s}_i\) be the largest non-unit singular value of \(P_i Q_{i-1}\). Let \(s_*\) be the largest non-unit singular value of \(T\). Then we have the bound
\begin{gather}
    \label{eq:projprod_singval_bound}
    s_*^2 \leq 1 - \prod_{i=2}^n (1 - \mathscr{s}_i^2)
\end{gather}
\end{theorem}
\begin{proof}
This is Theorem \ref{thm:singvals1} of the main text.
We will prove Equation \ref{eq:projprod_singval_bound} by induction. Let \(T_i = P_i ... P_1\) so that \(T = T_n\). Suppose that \(T_{i-1}\) satisfies equation \ref{eq:projprod_singval_bound}. We will prove that \(T_i = P_i T_{i-1}\) also satisfies \ref{eq:projprod_singval_bound}. \\
Let \(Y_i\) be the unit eigenspace of \(T_i\). Let \(v\) be a unit vector such that \(||T_i v || = s_*(T_i)\). We may take an orthogonal decomposition \(v = v_1 + v_2\), where \(v_1 \in Y_{i-1}\) and \(v_2 \perp Y_{i-1}\). We wish to compute
\[s_*(T_i) = ||P_i T_{i-1}v|| = ||P_i (v_1 + T_{i-1} v_2)||\]
Let \(\theta, \phi\) be the angles between \(T_i v\) and \(v_1, T_{i-1} v_2\), respectively. Note that \(||T_{i-1}v_2|| \leq s(T_{i-1}) ||v_2||\). Then \begin{gather}
||P_i (v_1 + T_{i-1} v_2)|| \leq ||v_1|| \cos \theta + s_*(T_{i-1}) ||v_2|| \cos \phi
\end{gather}
We next wish to optimize this bound over \(||v_1||, ||v_2||, \theta, \phi\) to obtain an unconditional bound. Our first constraint is \(||v_1||^2 + ||v_2||^2 = 1\), so after optimizing over the norms we find
\[||P_i (v_1 + T_{i-1}v_2)|| \leq \sqrt{\cos^2 \theta + s_*(T_{i-1})^2 \cos^2 \phi} \equiv f(\theta,\phi)\addtag\] 
Now we optimize over the angles. By their definition we must have \(0 \leq \theta \leq \frac{\pi}{2}\) and likewise for \(\phi\). Notice that \(f\) is a monotonically decreasing function of both \(\theta\) and \(\phi\) in this region, so the maximum will be attained somewhere on the boundary of the feasible set.

We also have some additional constraints. Since \(v_1 \perp M v_2\), we must have \(\theta + \phi \geq \frac{\pi}{2}\). And we know that the angle between \(v_1\) and \(P_i\) must be at least \(\cos^{-1} \mathscr{s}_i\) by the definition of \(\mathscr{s}_i\).

The Pareto frontier where both \(\theta\) and \(\phi\) are as small as possible lies along \(\theta + \phi = \frac{\pi}{2}\), so the optimum must be somewhere on this line. On this line we have \(\cos^2 \phi = 1 - \cos^2 \theta\), so
\[f(\theta) = \sqrt{(1 - s_*(T_{i-1})^2)\cos^2 \theta + s_*(T_{i-1})^2}\addtag\]
Since \(s_*(T_{i-1}) < 1\) this is again a decreasing function of \(\theta\), so it attains its maximum when \(\theta\) is minimized, i.e. \(\theta = \cos^{-1} \mathscr{s}_i\). We thus find
\[s_*(T_{i}) \leq \sqrt{(1 - s_*(T_{i-1})^2) \mathscr{s}_i^2 + s_*(T_{i-1})^2}\addtag\]
By assumption \(s_*(T_{i-1})^2 \leq 1 - \prod_{j=2}^{i-1} (1 - \mathscr{s}_j^2)\)
and so
\begin{align}
    s_*(T_{i})^2 &\leq \mathscr{s}_i^2 + (1-\mathscr{s}_i)^2 \left(1 - \prod_{j=2}^{i-1} (1 - \mathscr{s}_j^2)\right) \n 
    &\leq 1 - \prod_{j=2}^i (1 - \mathscr{s}_j^2)
\end{align}
completing the induction. Finally, for the base case, \(T_1\) is a single projector, so \(s_*(T_1) = 0\).
\end{proof}

The role of this theorem in our proof is analogous to that of the detectability lemma in e.g. ref.~\onlinecite{Haferkamp2022}. The detectability lemma is based on counting the number of projectors which do not commute, whereas here the \(\mathscr{s}_i\) in some sense quantify the amount of noncommutativeness.

\subsection{Other bounds on subleading singular values}
\begin{lemma}
\label{lemma:eigvsing}
Let \(s_*\) be the largest non-unit singular value of \(T\). Let \(\lambda_*\) be the non-unit eigenvalue of \(T\) with the largest magnitude. Then 
\[|\lambda_*| \leq s_*\addtag\]
\end{lemma}
\begin{proof}
We may write
\[s_* = \max_{\left\{v | v \perp Z \right\}} \frac{||T v||}{||v||}\addtag\]
Let \(v\) be an eigenvector of \(T\) with eigenvalue \(\lambda \ne 1\). From before \(v \perp Z\), and clearly \(\frac{||T v||}{||v||} = \lambda\), so we have \(s_* \geq \lambda\). The lemma follows immediately.
\end{proof}

\begin{theorem}
\label{thm:2layer_eigvsing}
Consider two projectors \(P_1, P_2\). Let \(\lambda_*\) be the largest non-unit eigenvalue of \(P_2 P_1\) and let \(s_*\) be the largest non-unit singular value. Then
\[s_* = \sqrt{\lambda_*}\addtag\]
\end{theorem}
\begin{proof}
    Any eigenvector $v_*$ corresponding to $\lambda_*$ must lie in the unit eigenspace of $P_2$. Otherwise, the output $P_2 P_1 v_*$ will not be parallel to $v_*$. Also, if we take the vector $w_* \equiv P_1 v_*$, by Lemma~\ref{lemma:eig_is_ortho} this vector is orthogonal to the unit eigensapce of $P_2 P_1$, and
    \begin{align}
        \frac{||P_2 P_1 w_*||^2}{||w_*||^2} &= \frac{v_*^T P_1 P_2^2 P_1 v_*}{v_*^T P_1 v_*}\n
        &= \frac{\lambda_*^2}{v_*^T (P_2 P_1 v_*)}\n
        &= \lambda_*
    \end{align}
    So the subleading singular value of $P_2 P_1$ is at least $\sqrt{\lambda_*}$.
    
    Now, take $w_*$ to be the subleading eigenvector of $(P_2 P_1)^\dagger (P_2 P_1) = P_1 P_2 P_1$. By definition the corresponding eigenvalue is $s_*^2$. Take $v_* = P_2 P_1 w_*$ (still orthogonal to the unit eigenspace of $P_2 P_1$ by Lemma~\ref{lemma:eig_is_ortho}). Then
    \begin{align}
        P_2 P_1 v_* &= P_2 P_1^2 P_2 P_1 w_*\n
        &= s_*^2 P_2 P_1 w_*\n
        &= s_*^2 v_*
    \end{align}
    We see that the subleading eigenvalue of $P_2 P_1$ is at least $s_*^2$. Combining these two inequalities establishes Theorem \ref{thm:2layer_eigvsing}.
\end{proof}

\begin{lemma}
\label{lemma:product_ssv_bound}
Let \(T_i, i\in \{1...k\}\) be a sequence of projector products that all share the same unit eigenspace. Let \(s_*^{(i)}\) be the largest non-unit singular value of \(T_i\). Then the largest non-unit singular value of \(\prod_{i=1}^k T_i\) is at most \(\prod_i s_*^{(i)}\).
\end{lemma}
\begin{proof}
Let \([M]_1\) denote the unit singular value space of \(M\), and let \(M_j = \prod_{i=j}^k T_i\)
 We can write the subleading singular value of \(M_k\) as
\[s(M_j) = \max_{v \perp \left[M_j\right]_1} \frac{||M_j v||}{||v||}\addtag\]
Since the \(T_i\) all share a unit eigenspace, the unit eigenspace of 
\(M_j\) is the same as that of each \(T_i\). But by Lemma \ref{lemma:leftright_eigenspace}, the unit eigenspaces and unit singular value spaces of both \(T_i\) and \(M_j\) are the same, so 
\[\left[M_j\right]_1 = [T_1]\addtag\]
We thus have
\[s(M_j) = \max_{v \perp [T]_1} \frac{||M_j v||}{||v||}\addtag\]
Furthermore, 
\[||M_j v|| = ||M_{j+1} T_j v)|| \leq S_j || M_{j+1} v||\addtag\]
and so
\[s(M_{j}) \leq s_*^{(j)} s(M_{j+1})\addtag\]
Induction then gives the desired bound.
\end{proof}

\subsection{Bounds on Frobenius norms}
\begin{lemma}
Let \(T\) be a projector product with largest non-unit singular value \(s_*\). Let \(d\) be the dimension of the space on which \(T\) acts, let \(m_1\) be the dimension of the unit eigenspace of \(T\), and let \(m_0\) be the dimension of the zero eigenspace. Then
    \[||T^k||_F^2 \leq m_1 + (d-m_1-m_0)s_*^2\addtag\]
\label{lemma:normbound_from_ssv}
\end{lemma}
\begin{proof}
Let \(\sigma_i\) be all the singular values of \(T\). We have \(\sigma_i = 1\) for \(i \in \{1...m_1\}\), \(\sigma_j \leq s_*\) for \(i \in \{m_1+1...d-m_0\}\), and \(\sigma_k = 0\) for \(i \{\in d-m_0+1...d\}\). We compute
\[||T||_F^2 = \sum_i \sigma_i^2 \leq \sum_{i=1}^m 1 + \sum_{i=m+1}^{d-m_0} s_*^2\addtag\]
from which the result follows immediately. We can further tighten the bound by replacing \(d\) with the number of nonzero singular values.
\end{proof}

\begin{theorem}
\label{thm:generic_normbound}
Let \(T_i, i\in \{1...k\}\) be a sequence of projector products that all have the same unit eigenspace. Let \(d\) be the dimension of the space on which \(T_i\) acts, let \(m_1\) be the dimension of the unit eigenspace of \(T_i\), and let \(m_0\) be the dimension of the zero eigenspace of any particular \(T_i\). Let \(s_*^{(i)}\) be the largest non-unit singular value of \(T_i\). Then
\begin{gather}
    \left| \left|\prod_i T_i \right| \right|_F^2 \leq m_1 + (d-m_1-m_0)\prod_{i=1}^k \left(s_*^{(i)}\right)^2
\end{gather}
\end{theorem}
\begin{proof}
This follows directly from the results of lemmas \ref{lemma:product_ssv_bound} and \ref{lemma:normbound_from_ssv}.
\end{proof}

\begin{theorem}
\label{thm:normbound_with_extras}
Let \(T_i, s_*^{(i)}, d, m_1, m_0\) be as in Theorem \ref{thm:generic_normbound}. Let \(R_1...R_n\) be a sequence of projector products such that the unit eigenspace of \(T_1\) is contained within the unit eigenspace of each \(R_i\). Let \(M\) be the product of all of the \(T_i\) and all of the \(R_i\) in any ordering. Then
\begin{gather}
    ||M||_F^2 \leq m_1 + (d-m_1-m_0)\prod_{i=1}^k \left(s_*^{(i)}\right)^2
\end{gather}
\end{theorem}
\begin{proof}
The \(R_i\) also preserve the unit eigenspace of \(T_i\) and are also norm-nonincreasing, so the proof of Lemma \ref{lemma:product_ssv_bound} still goes through. We can then again use Lemma \ref{lemma:normbound_from_ssv} to obtain the final formula.
\end{proof}

\begin{theorem}
Let \(T\) be a projector product with largest non-unit singular value \(s_*\). Let \(d, m_1, m_0\) be as before. Then
    \[||T^k||_F^2 \leq m_1 + (d-m_1-m_0)s_*^{2k}\addtag\]
\label{thm:periodic_normbound}
\end{theorem}
\begin{proof}
This is an immediate consequence of Theorem \ref{thm:generic_normbound} in the case where the \(T_i\) are all the same.
\end{proof}

\section{Proof of the 1D brickwork spectral gap for $t=2$}
\label{app:1d_t2_bound}

For a string of $N$ sites which we act a 1D brickwork on, we can consider the non-orthogonal basis
\begin{gather}
    |\vec{X} \rangle = \bigotimes_i |X_i\rangle_i
\end{gather}
where $X_i$ means the one of the two $k=2$ permutation states $|I\rangle$ (Identity) or $|S\rangle$ (Swap). This basis is complete on the image of any layer of the brickwork. We start with
\begin{lemma}
    Take $P_1$ to be the projector into the unit eigenspace of $T$. Then there exists a depth-independent constant $c_{X' X}$ such that
    \begin{gather}
        \langle \vec{X'} | T^\np |\vec{X}\rangle \leq \langle \vec{X'} | P_1 |\vec{X}\rangle + c_{X' X} \lambda_1^{\np}
    \end{gather} \label{lemma:1d_brickwork_bound}
    for every $\np > \frac{\log(N)}{-\log(\lambda_1)} + 1$, where $\lambda_1 = \left(\frac{2q}{q^2+1}\right)^4$
\end{lemma}
Proof: We use a domain wall trajectory approach\cite{Dalzell2022}. Each gate sends $|II\rangle \rightarrow |II\rangle$, $|SS\rangle \rightarrow |SS\rangle$, and the non-uniform $|IS\rangle, |SI\rangle \rightarrow q/(q^2+1)(|II\rangle + |SS\rangle)$. The point is that the transfer matrix sends each configuration into a sum of other configurations, depending on the position of the I,S domain walls in the system. If a gate in the transfer matrix crosses a domain wall (and all domain walls will be crossed by a gate each layer after the 1st layer), it either moves left with weight $q/(q^2+1)$ or right with weight $q/(q^2+1)$. So we can represent the transfer matrix of $\np$ layers as a series of domain wall trajectories with their accompanying weights. A domain wall trajectory is a sequence of $\vec{X}^j$'s such that $\vec{X}^0 = \vec{X}$ and $\vec{X}^{d} = \vec{X}'$. Specifically we're looking out for the domain walls in each layer $\vec{X}^j$, because the total number of domain walls never increases - each domain wall either moves around or annihilates with a neighbor, possibly eventually reaching the steady states $|I\rangle^N, |S\rangle^N$ with no domain walls. We have, for domain wall trajectories $\gamma$ with final state $F_\gamma$ and weight $w(\gamma)$,
\begin{gather}
    \langle \vec{X}' | T^\np |\vec{X}\rangle = \sum_{\gamma} w(\gamma) \langle \vec{X}' | F_\gamma \rangle \label{eq:trajectory_start}
\end{gather}
We can categorize each domain wall trajectory according to which domain walls annihilate before the end of the circuit and which remain to the end. We can separate the total domain wall trajectory into the annihilating (``closed'') domain walls $\gamma_c$ from the ``open'' domain wall trajectories $\gamma_o$ that stay to the end. The weight of the whole domain wall trajectory is just a product of the two components:
\begin{gather}
    w(\gamma) = w(\gamma_o) w(\gamma_c)
\end{gather}
Moreover
\begin{lemma}
    The final state $F_\gamma$ is determined solely by the open component of the trajectory.
\end{lemma}
\begin{proof}
    If we look at one of the surviving domain walls, the I/S values to the immediate left or right of the domain wall must remain the same through the whole trajectory (otherwise the domain wall would be annihilated). Then the rest of the final configuration can be determined by just flipping the sign when crossing every other domain wall. Because the signs to the immediate left and right of the domain wall were fixed by the original configuration $\vec{X}$, there is no global symmetry ambiguity either.
\end{proof}

With this we can separate the sum (\ref{eq:trajectory_start}) into a sum over the partitions $p$ that specify which domain walls are closed and which remain open. The partition in turn specifies the set of possible open and closed domain wall trajectories, $\mathcal{O}(p)$ and $\mathcal{C}(p)$:
\begin{align}
    \langle \vec{X}' | T^\np |\vec{X}\rangle &= \sum_p \sum_{\gamma_o \in \mathcal{O}(p)} \sum_{\gamma_c \in \mathcal{C}(p)} w(\gamma_o) w(\gamma_c) \langle \vec{X}' | F_\gamma \rangle\\
    &= \sum_p \sum_{\gamma_o \in \mathcal{O}(p)} w(\gamma_o)  \langle \vec{X}' | F_\gamma \rangle \left(\sum_{\gamma_c \in \mathcal{C}(p)} w(\gamma_c) \right)
\end{align}
Now we start to bound things. Firstly, 
\begin{lemma}
    $\sum_{\gamma_c \in \mathcal{C}(p)} w(\gamma_c)$ is bounded by the infinite depth limit, i.e. the sum of all the possible trajectories in an infinite depth circuit that start at the domain walls specified by $p$ but annihilate to either the uniform $|I\rangle$ or $|S\rangle$ states.
\end{lemma}
This is because $\mathcal{C}(p)$ is just a subset of all the possible closed trajectories, and increasing the depth on a trajectory that's already closed doesn't change its weight. The sum of weights in the infinite depth limit, in turn, has to be
\begin{gather}
    \lim_{\np \rightarrow \infty} \big|T^\np | \vec{X}_C(p)\rangle\big| = \big|P_1 | \vec{X}_C(p)\rangle\big|
\end{gather}
where $|\vec{X}_C(p)\rangle$ is the initial configuration given by only the closed domain walls, and $P_1$ is the unit eigenspace projector. That is, it's the component of $|\vec{X}_C(p)\rangle$ in the unit eigenspace of $T$, because the unit eigenspace is all that survives in the infinite depth limit. Let's call this infinite depth limit $W(p)$. 
We also have
\begin{lemma}
    \begin{gather}
        \sum_{\gamma_o \in \mathcal{O}(p)} w(\gamma_o) \leq \left(\frac{2q}{q^2+1}\right)^{2N_o (\np-1)}
    \end{gather}
    where $N_o$ is the number of open domain walls.
\end{lemma}
\begin{proof}
    We build up $\sum_{\gamma_o \in \mathcal{O}(p)} w(\gamma_o)$ layer by layer. At each layer beyond the first, each open trajectory $\gamma_o$ needs to move each of its $N_o$ domain walls. Each domain wall has at most 2 possible directions to move in (and could be less than 2 if other domain walls are blocking the way). No matter the direction, the domain wall acquires a weight $q/(q^2+1)$ by moving. So the total weight of all the possible trajectories created by adding a layer onto $\gamma_o$ is at most $w(\gamma_o)\left(\frac{2q}{q^2+1}\right)^{N_o}$. A transfer matrix is composed of two layers, so this is an extra factor of $\left(\frac{2q}{q^2+1}\right)^{2N_o}$ per transfer matrix applied. Extrapolating back to the first layer gives us the equation above.
\end{proof}

With these two inequalities, we have
\begin{align}
    \langle \vec{X}' | T^\np |\vec{X}\rangle &\leq \sum_p W(p) \sum_{\gamma_o \in \mathcal{O}(p)} w(\gamma_o)  \langle \vec{X}' | F_{\gamma_o} \rangle \\
    &\leq \sum_p W(p) x_F \sum_{\gamma_o \in \mathcal{O}(p)} w(\gamma_o) \\ 
    &\leq \sum_p W(p) x_F \lambda_1^{\half N_o (\np-1)}
\end{align}
For $x_F \equiv \max_{\gamma_o \in \mathcal{O}(p)} \langle \vec{X}' | F_{\gamma_o} \rangle$ and $\lambda_1 = \left(\frac{2q}{q^2+1}\right)^4$. The RHS is now a weighted sum of terms that are exponential in depth, with base $\lambda_1^{\half N_o}$ dependent on the number of open trajectories in the partition. The gentlest exponential is the single partition where $N_o=0$ and all the domain walls are closed - in that case $W(p)$ is the component of $|\vec{X}\rangle$ in the unit eigenspace and $\left[\max_{\gamma_o \in \mathcal{O}(p)} \langle \vec{X}' | F_{\gamma_o} \rangle\right]$ is the component of $|\vec{X'}\rangle$ in the unit eigenspace. This term is therefore bounded by $\lim_{\np \rightarrow \infty} \langle \vec{X}' | T^\np |\vec{X}\rangle$, which by Lemma~\ref{lm:cluster} is the same as $\langle \vec{X}' | P_1 |\vec{X}\rangle$

The rest of the partitions have $N_o \geq 2$ in periodic boundary conditions. Hence they decay at a rate $\lambda_1^{\np}$ or faster. Specifically, we have
\begin{gather}
    \langle \vec{X}' | T^\np |\vec{X}\rangle \leq \langle \vec{X}' | P_1 |\vec{X}\rangle + c_{X' X} \lambda_1^{\np}
\end{gather}
for some depth-independent constant $c_{X' X}$ (note that while $\max_{\gamma_o \in \mathcal{O}(p)} \langle \vec{X}' | F_{\gamma_o}\rangle$ is depth dependent, it's bounded above by 1).
Now we can prove that $\lambda_1 = \left(\frac{2q}{q^2+1}\right)^4$ genuinely is the subleading eigenstate of $T$, using
\begin{lemma}
    For a complete basis $|\vec{X}\rangle$, if $\langle \vec{X}' | T^\np | \vec{X} \rangle \leq \langle \vec{X}' | P_1 | \vec{X} \rangle + c_{X' X} \lambda_1^\np$, and $\langle \vec{X}' | T^\np | \vec{X} \rangle \geq \langle \vec{X}' | P_1 | \vec{X} \rangle$, then $T^\np$ has no eigenstate $|\psi_2\rangle$ with eigenvalue $\lambda_1 < \lambda_2 < 1$.\label{lemma:1dgapconfirmed}
\end{lemma}
\begin{proof}
Suppose there exists such an eigenstate $|\psi_2\rangle$. Then because $|\vec{X}\rangle$ is complete, there must exist some $\vec{X}$ which has a nonzero component of $|\psi_2\rangle$ in its eigenstate decomposition of $T$
\begin{gather}
    |\vec{X}\rangle = a_1 |\psi_1\rangle + a_2 |\psi_2\rangle + ..., \qquad a_2 \neq 0
\end{gather}
Moreover, there must exist some $\vec{X'}$ which has a nonzero overlap with $|\psi_2\rangle$, i.e. $\langle \vec{X'} | \psi_2 \rangle \neq 0$. Then we have
\begin{align}
    \langle \vec{X}' | T^\np |\vec{X}\rangle &= a_1 1^\np \langle \vec{X'} |\psi_1\rangle + a_2 \lambda_2^\np \langle \vec{X'} |\psi_2\rangle + ...\\
    &= 1^\np \langle \vec{X'} |P_1 | \vec{X}\rangle + a_2 \lambda_2^\np \langle \vec{X'} |\psi_2\rangle + ...
\end{align}
In particular, because $\lambda_2 > \lambda_1$, for any constant $c_{X' X}$ there must be some $\np$ for which $\langle \vec{X}' | T^\np | \vec{X} \rangle > \langle \vec{X}' | P_1 | \vec{X} \rangle + c_1 \lambda_1^\np$. This is a contradiction of our assumption, so no such $\lambda_2$ can exist. 

One caveat is that the eigenstate overlap $a_2 \langle \vec{X'} |\psi_2\rangle$ could be negative instead. Which is where the lower bound comes in - provided $\lambda_2$ is the largest subleading eigenvalue, there must also exist some $\np$ for which $a_2 \lambda_2^\np \langle \vec{X'} |\psi_2\rangle$ overtakes every nonunit term in the sum and the lower bound is violated instead. This lower bound is naturally satisfied in our case because $P_1$ is the infinite depth limit of $T^\np$, but $T$ is contractive, so adding more layers to $\langle \vec{X}' | T^\np |\vec{X}\rangle$ always decreases its norm. Because it's a sum of positive trajectories, it's also positive, so the value of $\langle \vec{X}' | T^\np |\vec{X}\rangle$ cannot increase as it goes to $P_1$. 
\end{proof}

Once we have $\lambda_1 = \left(\frac{2q}{q^2+1}\right)^4$, by Lemma~\ref{lemma:eigvsing}, $s_* = \left(\frac{2q}{q^2+1}\right)^2$. Therefore, our trace decays at a rate $C(q,2) = \frac{1}{2\log(1/s_*)} = \frac{1}{4\log \frac{q^2+1}{2q}}$.

\section{Mapping incomplete to complete layers for non-integer \(\sqrt{q}\)}
\label{app:irrational_q}

\newcommand{\splitx}[1]{{X_{\! \! \!\sqrt{#1}}}}

If \(\sqrt{q}\) is not an integer, we cannot draw Haar-random \(\sqrt{q} \times \sqrt{q}\) unitaries. There is thus no such thing as the 1D brickwork circuit ensemble. So how can we apply the site-splitting trick used in the proof of Theorem \ref{thm:periodic_incomplete}? Instead of defining \(s_\text{1D}(q)\) to be the subleading singular value of the transfer matrix corresponding to some underlying circuit ensemble, we will define the transfer matrix directly. 

The first step is to rephrase our site-splitting strategy from something done in the quantum circuits to something done at the level of transfer matrices. Consider a set of vector spaces labeled by a positive real-valued $r>1$:
\begin{gather}
    X_r = \text{span}\{\ket{\sigma}_{r}| \; \sigma \in S_t\}
\end{gather}
where each vector space is equipped with a basis $|\sigma\rangle_{X_r}$ and an inner product
\begin{gather}
    \braket{\sigma|\tau}_{r} = r^{|\sigma \tau^{-1}|}
\end{gather}
Here \(|\sigma|\) is the length of the cycle structure of \(\sigma\). Note that this inner product is positive semidefinite only for integer \(r\). 

Now, consider the mapping \(V: X_r \rightarrow \splitx{r}^{\otimes 2}\) defined on basis elements by
\[V\ket{\sigma}_{r} = \ket{\sigma}_{\sqrt{r}}^{\otimes 2}\addtag\]
We first show that $V$ is an isometry. Compute
\[\Big\langle V\ket{\sigma}, V \ket{\tau}\Big\rangle_{\sqrt r} = \braket{\sigma \sigma |\tau \tau}_{\sqrt{r}} = \left(\sqrt{r}^{|\sigma \tau^{-1}|}\right)^2\addtag\]
and
\[\braket{\sigma|\tau}_r = r^{|\sigma \tau^{-1}|}\addtag\]
which are the same. This implies that the restriction of the metric on $\splitx{r}^{\otimes 2}$ to the image of \(V\) is positive semidefinite (for integer $r$). 

For each $X_r$, let us define the $k$-site gate $G_k^{(X_r)}: X_r^{\otimes k} \rightarrow X_r^{\otimes k}$ as the following projector onto the span of $\{|\sigma\rangle_{r}^{\otimes k} | \; \sigma \in S_t\}$:
\begin{gather}
    G_k^{(X_r)} = \sum_{\sigma, \tau} |\tau\rangle^{\otimes k} \Wg(r^k)_{\sigma \tau} \langle \sigma|^{\otimes k}
\end{gather}
Where $\Wg(q)$ is the Moore-Penrose pseduoinverse of the metric $g(q)_{\sigma \tau} = \langle \sigma | \tau \rangle_r$. This formula reproduces the usual $k$-site gate when $r$ is an integer. 
We now observe the following:
\begin{lemma} \label{lemma:site_splitting_gate}
    If we replace every $G_k^{X_r}$ in the transfer matrix with $G_{2k}^{\splitx{r}}$, the singular values do not change.
\end{lemma}
\begin{proof} 
Let us define a map \(V^{\dagger}\) by
\begin{gather}
    V^\dagger = \sum_{\sigma \tau} \ket{\sigma}_r \Wg(r)_{\sigma \tau} \bra{\tau \tau}_{\sqrt{r}}
\end{gather}
so that
\[V^\dagger \ket{\tau \tau}_{\sqrt{r}} = \ket{\tau}_{r}\]
This map is an adjoint of \(V\) on the image of \(V\), i.e.
\[\bra{\sigma}_r V^\dagger \ket{\tau \tau}_{\sqrt{r}} = \bra{\tau \tau}_{\sqrt{r}} V \ket{\sigma}_r \addtag\]
for all  \(\tau,\sigma\). 

From the definition of \(V\), we see
\[G_k^{X_r} = V^{\dagger \otimes k} G_{2k}^{\splitx{r}} V^{\otimes k}\addtag\]
We may thus rewrite the transfer matrix $T$ by replacing \(G_k^{X_r}\) with the above expression for each gate. Furthermore, \(V V^\dagger = G_{1}^{\splitx{r}}\). Factors of \(G_1\) may be absorbed into \(G_k\) from either side, which means in particular
\[G_{2k}^{\splitx{r}} = (V V^{\dagger})^{\otimes k} G_{2k}^{\splitx{r}} = G_{2k}^{\splitx{r}} (VV^\dagger)^{\otimes k}\addtag\]
We may thus pair up the copies of \(V\) that appear on internal legs of the transfer matrix and absorb them into the adjacent \(G_{2k}^{\splitx{r}}\). The singular values of $T$ are the nonzero eigenvalues of $T^\dagger T$, so the copies of \(V\) which appear on input legs can be cycled to the output and cancelled against the corresponding copies of \(V^\dagger\) without changing the singular values.
\end{proof}

If we return to our original transfer matrix, we see that we can identify each site as a member of $X_q$, and each two-site gate as a copy of $G_2^{X_q}$. Moreover, we can freely apply copies of $G_1^{X_q}$ (the averaged one-site gate) to any site wherever we want, as it is just the identity on that vector space. 
Now we apply the isomorphism $V$ to map each site from an element of $X_q$ to an element of the doubled space $\splitx{q}^{\otimes 2}$ - converting each site into a pair of sites, or \textit{twinned sites}, in the process. From Lemma~\ref{lemma:site_splitting_gate}, we see that the transfer matrix can be rewritten in this new vector space, without changing any singular values, by replacing \(G_k^{X_q}\) with \( G_{2k}^{\splitx{q}}\) everywhere it appears.

In the particular case of an incomplete circuit, we obtain a transfer matrix consisting of \(G_1^{X_q}\) and \(G_2^{X_q}\). We then apply the isomorphism described above to map
\(G_1^{X_q} \rightarrow G_2^{\splitx{q}}\) and
\(G_2^{X_q} \rightarrow G_4^{\splitx{q}}\).

\begin{figure}
    \centering
    \begin{tikzpicture}
        \begin{scope}
            \node[anchor=north west,inner sep=0] (image_a) at (0,0)
            {\includegraphics[width=0.9\columnwidth]{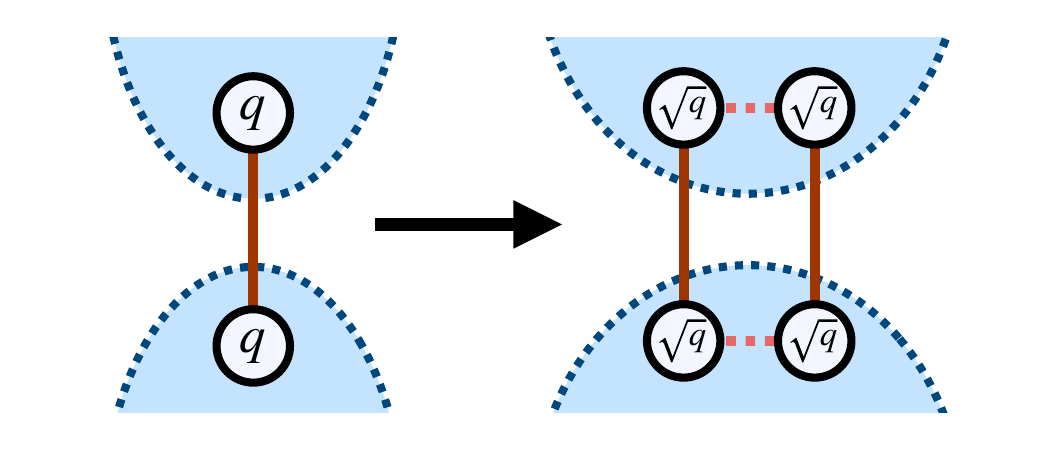}};
        \end{scope}
    \end{tikzpicture}
    \vspace*{-0.6cm}
    \caption{A two-site gate being split under the $X_q \rightarrow \splitx{q}^{\otimes 2}$ isomorphism. Each site on $X_q$ becomes a pair of twinned sites (pink dotted lines) on $\splitx{q}$. Note that the dimensionality of each site does not necessarily decrease - only the metric between the basis vectors changes. If we assume that the two-site gate spans two different clusters (blue dotted circles), then the layer-restricted SSV of a hyperedge connecting all four split sites is equal to two pairs of edges connecting non-twinned sites.}
    \label{fig:site_splitting}
\end{figure}

The former map sends $1$-site gates to $2$-site gates; these correspond to edges across two twinned sites. The first layer of the original circuit involved a one- or two-site gate acting on every site, so the split pairs sites are always joined back into the same cluster by the first layer. The new $4$-site gates are harder to express in the cluster-merging picture. However, we know that a mini-circuit of $2$-site gates $G_2^{\splitx{q}}$, applied in a way that connects all the sites consistently, will approach $G_4^{\splitx{q}}$ in the limit of infinite layers. In particular, we can replace each $4$-site gate $G_4^{\splitx{q}}[a_1, a_2, b_1, b_2]$ acting over sites $a_1...b_2$ (where $a_1$ and $a_2$ are twinned, and so are $b_1$ and $b_2$) with an arbitrarily large ``Jenga tower'' of the gates \small
\begin{gather*}
    \left(G_2^{\splitx{q}}[a_1, a_2] \otimes G_2^{\splitx{q}}[b_1, b_2] \right) \left(G_2^{\splitx{q}}[a_1, b_1] \otimes G_2^{\splitx{q}}[a_2, b_2] \right)
\end{gather*} \normalsize
repeated over and over again. In the cluster-merging picture, only the bottom layer of this tower is capable of joining distinct clusters together. The other layers either join twinned sites which were part of the same cluster to begin with, or are copies of the bottom layer. So every layer above the first has a completely disconnected cluster-merging graph, which corresponds to \(\mathscr{s} = 0\). These layers do not contribute to the right-hand side of Theorem~\ref{thm:singvals1}. So each $4$-site gate can be replaced with two $2$-site gates connecting non-twinned members together without increasing the subleading singular value of the transfer matrix. This process is illustrated in Fig.~\ref{fig:site_splitting}.

We see that we can replace the cluster-merging picture with one where every site is replaced by twins on $\splitx{q}$, and every $2$-site gate is replaced by a pair of $2$-site gates between non-twinned sites. 
Moreover, we can freely apply $2$-site gates between any twinned members where no gate was being originally applied. So all sites in the new picture have two-site gates acting on them; the layer is complete. Since all twinned sites belong to the same cluster, all clusters are of even size. These were the two conditions required to draw an Eulerian cycle on each connected component of the graph. We can now use Lemma~\ref{lemma:merge} to reduce this graph to a periodic 1D brickwork composed of gates $G_2^{\splitx{q}}$ acting on the space $\splitx{q}^{\otimes 2N}$. Any previous work which has found a bound for the $\sqrt{q}$ brickwork therefore imposes a bound on the layer-restricted subleading singular values of arbitrary cluster-merging graphs.

\section{Graph-splitting bounds for incomplete layers}
\label{app:treebound}
In this section we consider strategies for bounding cluster-merging graphs which originate from incomplete layers, without resorting to 1D brickwork bounds on non-integer $q$.

\subsection{Analytical bound in terms of node degree}
By removing edges according to Lemma~\ref{lemma:add_new_link}, we can reduce any connected cluster-merging graph into a spanning tree. Let $d$ be the maximum degree of the tree. 
To bound this tree's singular value, we use the following tool for incomplete graphs:
\begin{lemma}
    The edges of a cluster-merging graph can be split into separate layers without lowering the subleading singular values.
\end{lemma}
\begin{proof}
This follows from a simple reinterpretation of which gates belong to the same layer. We know that the gates corresponding to each edge in the cluster-merging graph commute with each other. Therefore, we are allowed to choose which gates to apply first. Splitting a layer merely means choosing a subset of gates to apply first, then choosing another subset to apply in the next new layer, and so on. This does not lower the subleading singular values because it does not change the gates at all.
\end{proof}
Note that splitting a layer will increase the block width $\ell$, which so our bound on the subleading singular value of the overall transfer matrix via Theorem \ref{thm:singvals} will get looser when a layer is split. Nonetheless, this will be useful for separating the degree-$d$ tree graph into reducible parts: 
\begin{theorem} \label{thm:treebound}
    A cluster merging graph of size $N$ with a spanning tree of maximum degree $d$ can be decomposed into at most 
    \begin{gather*}
        2 \,\min\!\left(\left\lceil d/2 \right\rceil, 6\right) \lceil\log_2(N)\rceil
    \end{gather*}
    layers, such that each layer is reducible to the 1D brickwork.
\end{theorem}
We will prove this statement starting with the similar
\begin{lemma}
    \label{lemma:treebound}
    A cluster merging graph of size $N$ with a  spanning tree of maximum degree $d$ can be decomposed into at most $\min\left(\lceil\frac{d}{2} \rceil, 6\right) \lceil \log_2(N) \rceil$ layers, such that each layer is made up of isolated strings.
\end{lemma}
Call $L(N)$ the largest possible number of layers a tree of size $N$ needs to be decomposed into to satisfy this. We will prove, by induction, $L(N) \leq \min\left(\lceil\frac{d}{2} \rceil, 6\right) \lceil \log_2(N) \rceil$. We have $L(2) = 1$ because two connected sites automatically form an isolated string. It remains to show
\begin{lemma}
    \label{lemma:tree_induction}
    $L(N) \leq L\left(\lfloor\frac{N}{2}\rfloor\right) + \min\left(\lceil\frac{d}{2} \rceil, 6\right)$
\end{lemma}
\textit{Proof.} Take an arbitrary root node $C_1$ in our tree and construct a path $C_1 C_2 ... C_r$ starting from that root. At each step $i-1$ in the path, we choose the child $C_i$ of parent $C_{i-1}$ with the most nodes in its subtree. Call the other child nodes $b_i^{(1)}, b_i^{(2)}, ..., b_i^{(g)}$, where $g \leq d-2$.

For a given path node $C_{i-1}$, each non-path child $b_i^{(1)}, b_i^{(2)}, ..., b_i^{(g)}$ has at most $\lfloor\frac{N}{2}\rfloor$ nodes in its subtree. We can therefore recursively decompose each subtree in parallel in at most $L\left(\lfloor\frac{N}{2}\rfloor\right)$ layers. After these layers, each subtree of $b_i^{(j)}$ has been combined into $b_i^{(j)}$ to form one cluster, $B_i^{(j)}$. 

We then need to combine all $g$ child nodes $B_i^{(j)}$ into $C_{i-1}$. To accomplish this in an efficient way, we use a method for efficient contraction of high-degree nodes:
\begin{figure}
    \centering
    \begin{tikzpicture}
        \begin{scope}
            \node[anchor=north west,inner sep=0] (image_a) at (0,0)
            {\includegraphics[width=0.99\columnwidth]{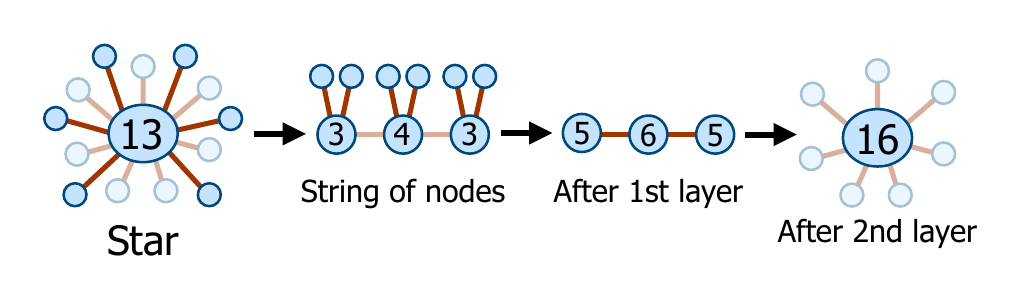}};
        \end{scope}
    \end{tikzpicture}
    \vspace*{-0.6cm}
    \caption{The star contraction of Lemma~\ref{lemma:star_contraction}. We start with a star with $g$ external edges, meaning its site count (the numerical label of the node) is at least $g$. In the first pair of layers \textit{(a)}, half of the edges are considered, then the star node is split into a string of nodes, each with 2 external edges. In the first layer, each node contracts with its 2 external edges in a 3-site string (Note that the site count on each node is just the minimum number of sites required to have the correct degree on the node. Extra sites can be placed in whatever node we want). In the 2nd layer, the resulting string is contracted back into a single node, leaving a star with half the edges remaining. The process is then repeated for the other half of the edges.}
    \label{fig:star_contraction}
\end{figure}
\begin{lemma}\label{lemma:star_contraction}
    We can combine a node with $g$ of its external neighbors using at most 4 layers of isolated strings.
\end{lemma}
\begin{proof}
To accomplish this, we combine our layer splitting process with a cluster splitting process. 
The original node has a degree of $g$, so it must contain at least $g$ sites. We first split the layer, with $\lfloor \frac{g}{2} \rfloor$ of the external neighbors on the bottom layer. This gives us a node with $g$ sites and $\lfloor \frac{g}{2} \rfloor$ edges. We then use Lemma \ref{lemma:splitlink} to split this node into a string of nodes such that each node connects to $2$ of the external neighbors (as shown by the first step in Fig. \ref{fig:star_contraction}). Each node therefore has $2$ external neighbors and at most $2$ new internal edges connecting it to the other nodes in the string, so can be made with $4$ sites. The exceptions are the end nodes which need one less site, and if $\lfloor \frac{g}{2} \rfloor$ is odd, one node will have one less neighbor and site. In total, at most $4\lceil \frac{1}{2} \lfloor \frac{g}{2} \rfloor \rceil - 2 \leq g$ sites are required, so there are enough unoccupied sites in the original node to add the necessary new edges for this splitting.

Once we have split the original node in this way, we spend two layers to combine all the neighbors and nodes together. In the first layer, each node in the string combines with its two external neighbors, as together they form a string of length 3 (Fig.~\ref{fig:star_contraction}, second step). With all the neighbors absorbed, the string can be recombined into the original node in the second layer (Fig.~\ref{fig:star_contraction}, third step).

We then spend two more layers to contract the remaining $g - \lfloor \frac{g}{2} \rfloor$ neighbors. This is either $\lfloor \frac{g}{2} \rfloor$ or $\lfloor \frac{g}{2} \rfloor+1$ depending on the parity of $g$. Even if it's the latter, we will have enough sites to contract all the neighbors, due to the $\geq \lfloor \frac{g}{2} \rfloor$ extra sites the root node gained by absorbing its previous neighbors. Therefore all neighbors can be contracted in 4 layers.
\end{proof}

Note that Lemma~\ref{lemma:star_contraction} is inefficient if the number of leaves is 6 or less. This is because we can also contract $g$ leaves in $\lceil \frac{g}{2} \rceil$ layers, by selecting one or two leaves in each layer and combining with $C_{i-1}$ into a string of length $2$ or $3$. So the number of layers we have to spend is $\min(\lceil \frac{g}{2} \rceil, 4) \leq \min(\lceil \frac{d}{2} \rceil-1, 4)$ overall.

By performing all of these operations in parallel for each path node, we have spent at most $L\left(\lfloor\frac{N}{2}\rfloor\right) + \min(\lceil \frac{d}{2} \rceil-1, 5) $ layers to combine each path node's non-path children into the path node. After these layers, the cluster-merging picture of the next layer only consists of the path nodes, which can be combined together with a single, extra, layer. This completes the proof of Lemma \ref{lemma:tree_induction}. \qed

We have seen that we can decompose the size-$N$ cluster into $L(N)$ layers such that the cluster merging graph of each layer consists of isolated strings. Each string can now be split into an open 1D brickwork, as long as each internal node is of even size (endpoint nodes can be either odd or even), using Lemma~\ref{lemma:merge}. We will call such strings \textit{brickwork-compatible}. For strings containing nodes which aren't even and aren't on the ends, we can use:
\begin{lemma} \label{lemma:odd_node_splitting}
    Any layer of isolated strings can be split into two layers of brickwork-compatible strings.
\end{lemma}
\begin{proof}
We wish to split a string into two layers such that both layers have cluster-merging graphs consisting entirely of brickwork-compatible strings. Let \(k\) be the number of odd-sized nodes in the string and number only the odd nodes \(1...k\).

First suppose \(k\) is even. Remove the left-hand edge, if any, of odd nodes \(1,3...k-1\) and the right-hand edge of odd nodes \(2,4...k\). Now we have split the string into at least \(\frac{k}{2}\) and at most \(\frac{k}{2} + 2\) substrings, each of which either has odd nodes only at both endpoints, or has no odd nodes at all. Since the odd nodes were paired up, contracting the first layer gives only even nodes for the second layer. So both layers are brickwork-compatible.

Now suppose \(k\) is odd. Remove the left-hand edge, if any, of odd nodes \(1,3...k-1\) and the right-hand edge of odd nodes \(2,4...k-1\). Every substring has \(2\) odd endpoints except the first, which may have \(0\), and the last, which has exactly \(1\). After contracting, the second layer has all even nodes except for one. But this one odd node is the right endpoint of the second-layer string, so the second layer is still brickwork-compatible.
\end{proof}

Lemma~\ref{lemma:odd_node_splitting} shows that we can make a set of layers of isolated strings brickwork-compatible with a splitting scheme that at most doubles the number of layers. Combining this with Lemma \ref{lemma:treebound} completes the proof of Theorem~\ref{thm:treebound}.

\subsection{Log log bound on arbitrary graphs}

\begin{theorem}\label{thm:loglog_bound}
    A cluster merging graph of size $N$ can be decomposed into at most
    \begin{gather*}
        8\lceil \log_2 \lfloor \log_2 (N+1) \rfloor \rceil +2
    \end{gather*}
    layers, such that each layer is reducible to the 1D brickwork.
\end{theorem}
Given Lemma~\ref{lemma:odd_node_splitting}, it is sufficient to prove we can decompose the graph into
\begin{gather*}
    4\lceil \log_2 \lfloor \log_2 (N+1) \rfloor \rceil +1
\end{gather*}
layers of isolated strings. 

We begin with a tree of at most \(N\) sites (Fig.~\ref{fig:tree_reduction}a, left). The first layer of edges we apply will contract substrings of the tree such that the resulting second-layer tree (Fig.~\ref{fig:tree_reduction}a, right) has depth \(O(\log N)\). Choose any root node, which also specifies a direction in the tree from its root down to the leaves. Then label each node by its \textit{subtree weight} - the number of nodes in the subtree starting from that node. The root node has subtree weight $N$, leaves have subtree weight $1$, and so on. We add to the first layer the \textit{maximally weighted path} starting from the root node, i.e. a path from root to leaf that always selects the child node with the highest subtree weight. For every node which neighbors this path (excluding nodes which are part of the path), we take its subtree, add the maximally weighted path of that subtree to the first layer, and repeat the process recursively. After the end of this process, each node is part of exactly one maximally weighted path, so the first layer consists of disconnected strings. We can also label each maximally weighted path by a subtree weight, which in this case is the subtree weight of the node at the start of the path.

After contracting all the paths added to the first layer this way, we are left with a tree where each node corresponds to one of the maximally weighted paths. We make the following claim:

\begin{figure}
    \centering
    \begin{tikzpicture}
        \begin{scope}
            \node[anchor=north west,inner sep=0] (image_a) at (-0.6,0)
            {\includegraphics[width=0.99\columnwidth]{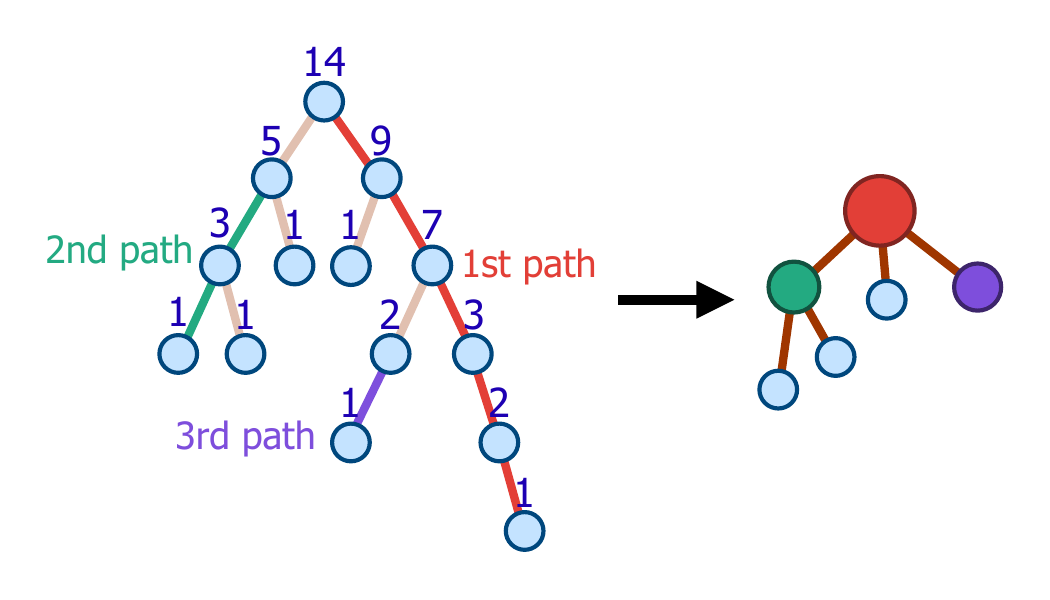}};
            \node [anchor=north west] (note) at (-0.45,0) {\small{\textbf{a)}}};
        \end{scope}
        \begin{scope}
            \node[anchor=north west,inner sep=0] (image_a) at (0,-4.3)
            {\includegraphics[width=0.9\columnwidth]{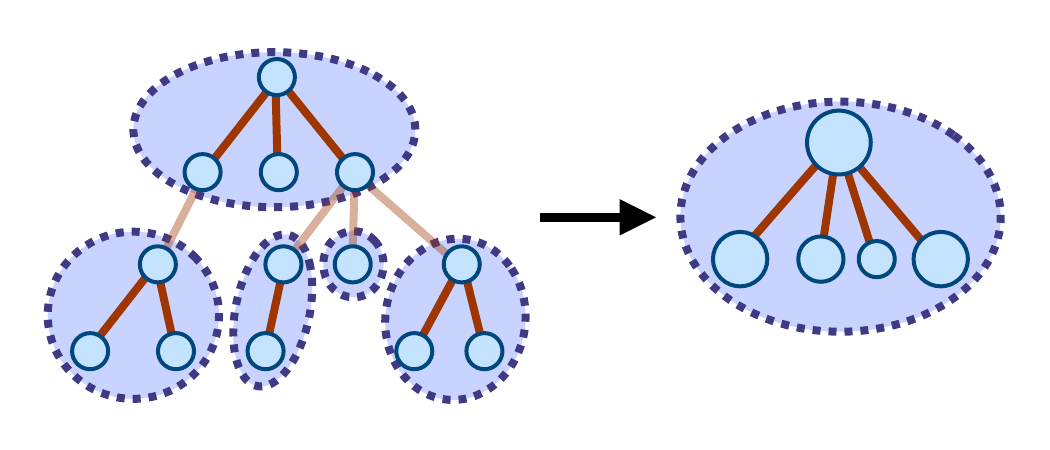}};
            \node [anchor=north west] (note) at (-0.45,-4.6) {\small{\textbf{b)}}};
        \end{scope}
    \end{tikzpicture}
    \vspace*{-1cm}
    \caption{Algorithm for reducing a general tree in $O(\log \log N)$ layers. (a): The first layer of the algorithm, where we contract the maximally weighted paths. Each path chooses the child node with the highest subtree weights (blue labels above each node), and starts from non-path children of nodes in higher paths. (b): Contraction of a balanced tree into a point in $O(\log \log N)$ steps. At each step, the nodes at depths $2k$ and $2k+1$ are merged together through the star contraction algorithm (each star is indicated by a blue dotted circle).}
    \label{fig:tree_reduction}
\end{figure}

\begin{lemma}
    The height of the new tree is at most $\lfloor \log_2(N+1) \rfloor \rceil-1$. 
\end{lemma}
\begin{proof}
    Suppose we start from a node of the new tree and jump down to one of its children. This corresponds to jumping from one maximally weighted path on the original tree to one of the paths below it. 
    
    Going back to the original tree, consider a maximally weighted path of a particular single-root subtree consisting of $n$ nodes. The subtree weight of the root is, by definition, $n$, and so is the subtree weight of the path itself. Then any node which is one edge away from the maximally weighted path (i.e. it is not part of the maximally weighted path, but its parent is) has a subtree weight bounded above by $\lfloor \frac{n-1}{2}\rfloor$. We can prove this by contradiction - it cannot be the root node, and the sum of the subtree weights of it and all its sibling nodes must be $\leq n-1$. So if it had a subtree weight greater than $\frac{n-1}{2}$, it would've had a larger subtree than all of its siblings, and therefore would be part of the maximally weighted path after all. 

    Every time we move from a maximally weighted path to a lower one, our subtree weight roughly halves; specifically, it goes from $n$ to at most $\lfloor \frac{n-1}{2}\rfloor$. If we repeat this process, we run out of nodes after $\lfloor \log_2(n+1) \rfloor-1$ steps. Hence, the height of the whole second-layer tree must be bounded by $\lfloor \log_2(N+1) \rfloor \rceil-1$.
\end{proof}

Now we reduce the second-layer tree with the following theorem:
\begin{lemma}
    A tree of height $h$ can be decomposed into $4\lceil \log_2 (h+1) \rceil$ layers of isolated strings
\end{lemma}
\begin{proof}
    We decomposed the tree recursively. At each step, we take all the nodes which are at an even depth (including the root node at depth 0). We can then use Lemma~\ref{lemma:star_contraction} on each even-depth node to absorb all their children in parallel, in at most $4$ layers (Fig.~\ref{fig:tree_reduction}b, left). This way, after $4$ layers, we have removed all the nodes which were at an odd depth. Hence, we are left with a tree of height $\lceil \frac{h+1}{2} \rceil-1$ (Fig.~\ref{fig:tree_reduction}b, right). We repeat this process until we are left with a tree of height $0$, i.e. just the root node. This takes $\lceil \log_2(h+1) \rceil$ steps, or $4\lceil \log_2(h+1) \rceil$ layers, in total. 
\end{proof}

Hence, we can spend $1$ layer to reduce an arbitrary tree to a tree of height $\lfloor \log_2(N+1)\rfloor-1$, then spend $4\lceil \lfloor \log_2(N+1)\rfloor \rceil$ layers to reduce the new tree to a single node, which proves Theorem~\ref{thm:loglog_bound}.

\section{Numerical evidence for conjectures}
\subsection{Evidence for \(t\)-independence of \(s_\text{1D}\)}
\label{app:brickwork_t_independence}
Figures \ref{fig:periodic_brickwork_tdependence} and \ref{fig:open_brickwork_tdependence} present numerical evidence for Equations \ref{eq:C_exact_periodic} and \ref{eq:C_exact_open}, respectively. We calculate \(s_\text{1D}(N,q,t)\) with both open and closed boundary conditions, for \(q = 2\) and several values of \(N\) and \(t\). 

\begin{figure}[h]
    \centering
    \begin{tikzpicture}
        \begin{scope}
            \node[anchor=north west,inner sep=0] (image_a) at (0.6,0)
            {\includegraphics[width=0.7\columnwidth]{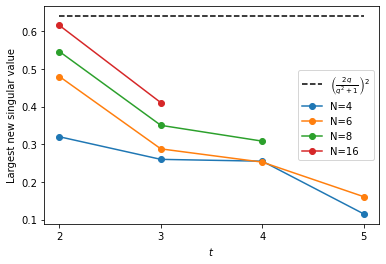}};
        \end{scope}
    \end{tikzpicture}
    \caption{Largest new singular values at each $t$ for various $N$. The circuit is a 1D brickwork with periodic boundary conditions and $q=2$. Singular values are calculating using a variant of the Lanczos algorithm. All obey the conjectured upper bound.}
    \label{fig:periodic_brickwork_tdependence}
\end{figure}

\begin{figure}[h]
    \centering
    \begin{tikzpicture}
        \begin{scope}
            \node[anchor=north west,inner sep=0] (image_a) at (0.6,0)
            {\includegraphics[width=0.7\columnwidth]{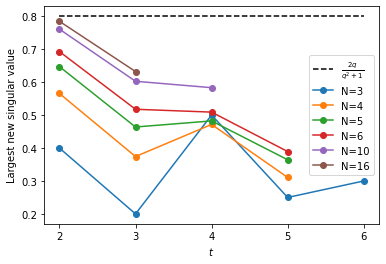}};
            \node [anchor=north west] (note) at (-0.45,0) {\small{\textbf{a)}}};
        \end{scope}
        \begin{scope}
            \node[anchor=north west,inner sep=0] (image_a) at (0,-4.0)
            {\includegraphics[width=0.8\columnwidth]{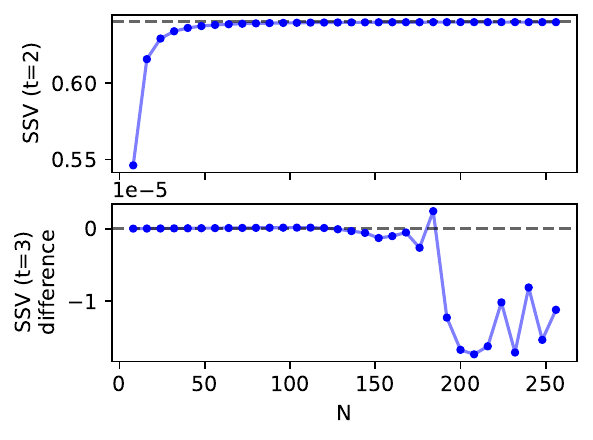}};
            \node [anchor=north west] (note) at (-0.45,-4.0) {\small{\textbf{b)}}};
        \end{scope}
    \end{tikzpicture}
    \caption{(a): Largest new singular values at each $t$ for various $N$. The circuit is a 1D brickwork with open boundary conditions and $q=2$. Singular values are calculating using a variant of the Lanczos algorithm. All obey the conjectured upper bound. (b): DMRG results for the open 1D brickwork subleading singular value $s_{\text{1D,open}}(N,q,t=2)$ vs. $N$ \textit{(upper)} and its difference from the $t=3$ SSV \textit{(lower)}. Note the vertical axis scale of $10^{-5}$ in the lower panel.}
    \label{fig:open_brickwork_tdependence}
\end{figure}

\subsubsection{Lanczos method} For small $N$, we use a variant of the Lanczos algorithm on $T^\dagger T$. In order to find new eigenvalues, we note that the gate has a block-triangular structure: The subspace of states involving only permutations on the first $t-1$ copies is closed under the action of the gate. This subspace cannot contain any new eigenvalues at $t$, since all of its eigenvectors are isomorphic to a state that appeared at $t-1$. 

Furthermore, since the gate operators commute with the right-action \(R(\tau) \ket{\sigma_1, \sigma_2 ...} = \ket{\sigma_1 \tau, \sigma_2 \tau ...} \), the subspace of states related to a state in the $t-1$ subspace by the right-action of any permutation also cannot contain a new eigenvalue. An analogous argument goes for left-action. 

Finally, the degeneracy of the metric for $q < t$ implies that some states which appear to be in the "new" subspace are actually equivalent to states in the $t-1$ subspace. These states also cannot contain a new eigenvalue. 

In order to find the largest new eigenvalue, we modify the Lanczos algorithm by adding a step which projects out all states in the subspaces described above. This can be done efficiently using a power iteration method on the product of the subspace projectors. 

All singular values found obey the conjectured bound. In fact, the largest singular value for any given circuit almost always appears at $t=2$. The sole exception observed is with open boundary conditions, $N=3$, and $q=2$, where the largest singular value found appears at $t=4$ (Fig.~\ref{fig:open_brickwork_tdependence}a). Note that both conjectures are already known to hold at $t = 2$.

\subsubsection{DMRG method} Due to the transfer matrix's natural representation as a tensor network, and hence a Matrix Product Operator, we can also use the Density Matrix Renormalization Group (DMRG) algorithm to approximate eigenvectors for much larger $N$. 

The accuracy of DMRG depends on the entanglement structure of the true eigenvector, quantified by its bond dimension. Here we choose a singular value cutoff of $10^{-12}$ and a maximum bond dimension of 800, which is well above the t=2 or t=3 required bond dimension of $\approx 12$.

\begin{figure}[h]
    \centering
    \begin{tikzpicture}
        \begin{scope}
            \node[anchor=north west,inner sep=0] (image_a) at (0,0)
            {\includegraphics[width=0.9\columnwidth]{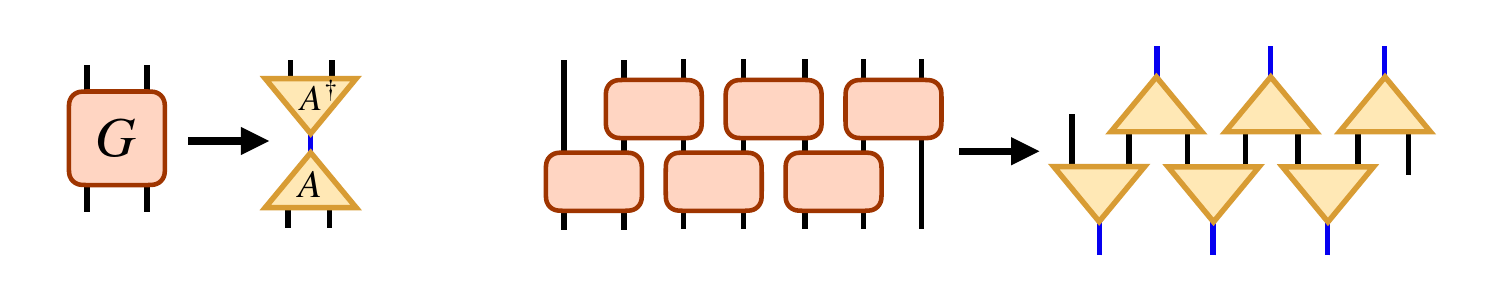}};
            \node [anchor=north west] (note) at (-0.15,-0.2) {\small{\textbf{a)}}};
            \node [anchor=north west] (note) at (2.25,-0.2) {\small{\textbf{b)}}};
        \end{scope}
    \end{tikzpicture}
    \caption{(a): Decomposition of a single gate into two halves using SVD. The inner index \textit{(blue)} can be treated as effective physical indices with dimension $t!$. (b): Decomposition of two layers of a 1D brickwork using this splitting process. The bottom-most and top-most halves of the original brickwork are set aside for now, making this equivalent to one brickwork layer in thickness. Through this process, the physical space has reduced from $N$ legs of size $q^{2t}$ to $\big\lceil\frac{N}{2}\big\rceil$ legs of size $t!$.}
    \label{fig:brickwork_site_splitting}
\end{figure}

Both methods are computationally intractable in the original representation of \(G\), which has four legs each of dimension \(q^{2t}\). However, in this representation \(G\) is very sparse. For the brickwork circuit, we can use singular value decompositions to compress \(G\) to a tensor with three legs, each of dimension \(t!\). The resulting tensor is an orthogonal projector from \(S_t\) to \(S_t^2\) under the metric induced by the Weingarten functions. This compression leaves the nonzero singular values of the whole circuit unchanged. The compressed circuit is illustrated in Figure \ref{fig:brickwork_site_splitting}.

\subsection{Evidence for Conjecture \ref{conjecture:connection_count}}
\label{app:connection_numerics}
To search for violations of Conjecture \ref{conjecture:connection_count}, we used a simulated annealing algorithm which attempts to maximize the subleading singular value over the set of architectures on a fixed number of sites. Figure \ref{fig:simanneal_data} shows the results.

\begin{figure}[h]
    \centering
    \begin{tikzpicture}
        \begin{scope}
            \node[anchor=north west,inner sep=0] (image_a) at (0,0)
            {\includegraphics[width=0.9\columnwidth]{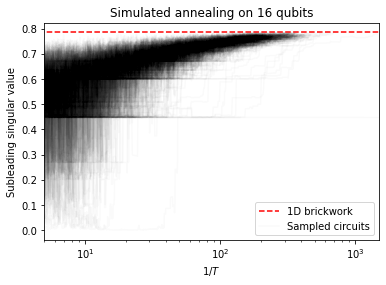}};
            \node [anchor=north west] (note) at (-0.45,0) {\small{\textbf{a)}}};
        \end{scope}
        \begin{scope}
            \node[anchor=north west,inner sep=0] (image_a) at (0,-5.5)
            {\includegraphics[width=0.9\columnwidth]{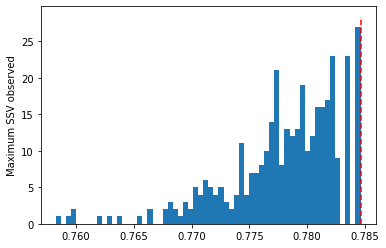}};
            \node [anchor=north west] (note) at (-0.45,-5.5) {\small{\textbf{b)}}};
        \end{scope}
    \end{tikzpicture}
    \caption{(a) Subleading singular values by inverse temperature during 437 runs of the simulated annealing process. Here $N = 16$, $q = 2$, and $t=2$. (b) Maximum SSVs attained by each each run.}
    \label{fig:simanneal_data}
\end{figure}

Each proposed move was a set of $n$ edge additions or deletions, with $n$ drawn from a geometric distribution with mean 1. Additions and deletions were equally likely, with additions drawn uniformly from the all-to-all graph and deletions drawn uniformly from the set of existing edges.

The objective function was the 3\textsuperscript{rd}-largest singular value, with starting temperature set automatically based on the distribution of singular values over a small sample of connected architectures with edges drawn i.i.d. from the all-to-all graph. An exponential multiplicative cooling schedule was used to cool to the final temperature over 5000 iterations.

\section{Tighter bounds for brickworks of arbitrary dimension}
\label{app:brickworks}

For specific, well-connected architectures, we can use the cluster-merging picture to obtain a tighter bound than that of 
Theorem~\ref{thm:periodic_complete}. One example is the generalized brickwork architectures on any dimension. We make the following claim:
\begin{lemma}
    For $t=2$, the effective connection depth $\ell$ on a generalized $d$-dimensional brickwork architecture on $L^d$ sites is at most $2+o_L(1)$. That is, the subleading singular value
    \[s \leq 1-(1-s_\text{1D})^2 + o_L(1) \addtag\]
\end{lemma}
In other words, we have taken the $2d$ layers of a $d$-dimensional brickwork's periodic block, and have effectively removed all but two of them from consideration in Theorem~\ref{thm:singvals1}.

Each of the $2d$ layers in a brickwork's periodic block chooses one of the $d$ directions (horizontal, vertical, etc.) and one of two parities (odd or even). 
The cluster merging picture of these layers comes in three stages. In the first stage, each layer has a different direction from all the layers below it. These layers combine finite (i.e. $O(1)$) size clusters into other finite size clusters. Each cluster after $m$ such layers is an $m$ dimensional hypercube over $2^m$ sites. Depending on the order of the layers, this first stage can have any number of layers from $1$ to $d$. The second stage begins once a layer has the same direction as a previous layer below it. In this stage, the hypercube clusters get strung together along this direction, $L/2$ at a time, forming loops that can be reduced to periodic 1D brickworks. This stage consists of only one layer.

After this stage, each cluster is of size $\Omega(L)$. Subsequent layers will connect these clusters either into pairs (if their same-direction counterpart wasn't applied yet) or $L/2$-length loops (if it was applied). Either way, each cluster will have $\Omega(L)$ connections with each neighbor in the graph. This effectively creates a 1D brickwork with internal dimension $\tilde{q} = q^{\Omega(L)}$. Since the 1D brickwork singular value is $\frac{2\tilde{q}}{\tilde{q}^2+1} = O(\tilde{q}^{-1})$, all of these layers will have layer-restricted singular values which decay exponentially in $L$, hence will contribute a vanishing amount $o_L(1)$ to Equation~\ref{eq:singvals1}. 

We have a second stage with one layer that reduces to the periodic 1D brickwork, and a third stage with vanishing contribution. Therefore, it remains to show that the combined subleading singular value of the first stage layers is bounded by the 1D brickwork. In other words, we want to bound the subleading singular value of an arbitrary hypercube by the 1D brickwork. 

The subleading singular value of the $2$D hypercube (i.e. a square of four sites) is $s_2 = \frac{2q^2}{(q^2+1)^2}$. We can use the cluster merging picture to bound hypercubes of higher dimension. At layer $m$, our clusters are size $2^{m-1}$, and we are joining them together in pairs, with $2^{m-1}$ connections per pair. Since we only have two clusters in each connected section of our graph, finding the layer-restricted singular value $\mathscr{s}_m$ just amounts to an optimization over $4$ basis states. We have
\begin{align}
    \mathscr{s}_{m} &= \frac{1}{(1-q^{-2^m})^2}\sqrt{2(1+q^{-2^m})\left[\frac{2}{q^2+1}\right]^{2^{m-1}}  - 8q^{-2^m}}\n
    &\leq \sqrt{2}\left(\frac{16}{15}\right)\left(\frac{2}{q^2}\right)^{2^{m-2}}\n
    &\leq \sqrt{2}\left(\frac{16}{15}\right)\left(\frac{2}{q^2}\right)^{2(m-2)}
\end{align}
where we bound the superexponential decay in $m$ by an exponential for all $m \geq 3$. Then any $d\geq 3$-dimensional hypercube $s_d$ has singular values bounded by
\begin{align}
    s_d^2 &\leq 1 - (1-s_2^2) \prod_{m=3}^d (1-\mathscr{s}_m^2) \n
    &\leq 1 - (1-s_2^2) \exp \left[-\sum_{m=3}^d \mathscr{s}_m^2\right] \n
    &\leq 1 - (1-s_2^2) \exp \left[-2\left(\frac{16}{15}\right)^2 \sum_{m=3}^d \left(\frac{2}{q^2}\right)^{4(m-2)}\right] \n
    &\leq 1 - \left(1-\frac{4q^4}{(q^4+1)^2}\right) \exp \left[-2\left(\frac{16}{15}\right)^2 \frac{16}{q^8-16}\right] 
\end{align}
This decays more rapidly in $q$ than the 1D brickwork, and at $q=2$ we get the bound $s_d(2) \leq 0.478 < s_\text{1D}(2)$. Therefore, the hypercube singular value is below the 1D brickwork singular value for all $q \geq 2$. 

A more thorough optimization over the size $M$ of the first stage will probably get an overall singular value bound that is much closer to $s_\text{1D}$. This is because the size, and connectivity, of the clusters in the second stage is $q^{2^M}$, so it would probably produce a singular value that is close to $O(q^{-2^M})$. In other words, we aren't allowed to make the first stage that deep, without reducing the layer-restricted singular value of the second stage.

\end{appendices}
\end{document}